\documentclass[preprint,12pt]{elsarticle}

\journal{Information and Computation}

\usepackage{amsmath}
\usepackage{amssymb}
\usepackage{mathtools}

\usepackage{hyperref}

\usepackage{subfigure}
\usepackage{wrapfig}

\usepackage{tikz}
\usetikzlibrary{arrows,automata,positioning}

\usepackage[textsize=footnotesize]{todonotes}

\renewcommand{\epsilon}{\varepsilon}

\newcommand{\myquot}[1]{``#1''}

\newcommand{\set}[1]{{\{#1\}}}
\newcommand{\Set}[1]{{\left\{#1\right\}}}
\newcommand{\tuple}[1]{{\langle#1\rangle}}
\renewcommand{\models}{\vDash}

\newcommand{\bigo}{\mathcal{O}}
\newcommand{\pow}[1]{2^{#1}}

\newcommand{\bound}{b}
\newcommand{\card}[1]{\left| {#1} \right|}

\newcommand{\nats}{\mathbb{N}}
\newcommand{\bool}{\mathbb{B}}

\newcommand{\strat}[2]{(2^{#1})^* \rightarrow 2^{#2}}
\newcommand{\proj}{\mathrm{proj}}
\newcommand{\wide}{\mathrm{wide}}
\newcommand{\fun}[2]{#1 \rightarrow #2}

\newcommand{\tsys}{\mathcal{S}} 

\newcommand{\ubuchi}{\mathcal{U}}
\newcommand{\rungraph}{\mathcal{G}}
\newcommand{\nbw}{\mathcal{N}}
\newcommand{\ucw}{\mathcal{U}}
\newcommand{\uct}{\mathcal{U}_S}

\newcommand{\bwin}{B}
\newcommand{\Bwin}{\mathcal{B}}
\newcommand{\cobwin}{\reflectbox{$B$}}

\newcommand{\true}{\mathbf{tt}}
\newcommand{\false}{\mathbf{ff}}
\renewcommand{\implies}{\mathbin{\rightarrow}}
\DeclareMathOperator\F{\mathbf{F}}
\DeclareMathOperator\G{\mathbf{G}}
\DeclareMathOperator\GF{\mathbf{GF}}
\DeclareMathOperator\FG{\mathbf{FG}}
\DeclareMathOperator\GFp{\mathbf{GF_P}}
\DeclareMathOperator\Fp{\mathbf{F_P}}
\newcommand{\U}{\mathbin{\mathbf{U}}}
\DeclareMathOperator\X{\mathbf{X}}
\newcommand{\R}{\mathbin{\mathbf{R}}}

\newcommand{\var}{\mathrm{var}}
\newcommand{\varG}{\mathrm{var}_{\mathbf{G}}}
\newcommand{\varF}{\mathrm{var}_{\mathbf{F}}}
\newcommand{\cl}{\mathrm{cl}}

\newcommand{\Var}{\mathcal{V}}
\newcommand{\rel}{\mathit{rel}}
\newcommand{\alt}{\mathit{alt}}

\newcommand{\ap}{\mathrm{AP}} 

\newcommand{\sched}{\mathit{sched}} 
\newcommand{\Sched}{\mathit{Sched}}

\newcommand{\ltl}{\mathrm{LTL}}

\newcommand{\pltl}{\mathrm{PLTL}}
\newcommand{\prompt}{\mathrm{PROMPT}$\textendash$\ltl}
\newcommand{\pldl}{\mathrm{PLDL}}
\newcommand{\ldl}{\mathrm{LDL}}
\newcommand{\ldlt}{\mathrm{LDL_{cp}}}

\newcommand{\arch}{\mathcal{A}}
\newcommand{\penv}{{p_\mathit{env}}}
\newcommand{\pminus}{P^-}
\newcommand{\distprod}{\otimes}
\newcommand{\Distprod}{\bigotimes}

\newcommand{\bgraph}{G}   
\newcommand{\cbgraph}{G}  

\newcommand{\nlogspace}{\textsc{NLogSpace}}
\newcommand{\pspace}{\textsc{PSpace}}
\newcommand{\twoexp}{\textsc{2ExpTime}}

\newcommand{\halfthinspace}{{\kern .08333em}}

\newcommand{\conc}{\,;}

\newcommand{\ddiamond}[1]{\langle\/ #1 \/\rangle\,}
\newcommand{\bbox}[1]{[\halfthinspace#1\halfthinspace]\,}

\newcommand{\ddiamondle}[2]{{\langle\/ #1 \/\rangle}_{\!\le #2}\,}
\newcommand{\bboxle}[2]{{[\halfthinspace#1\halfthinspace]}_{\le #2}\,}

\newcommand{\Rexp}{\mathcal{R}}

\newcommand{\vardiamond}{\mathrm{var}_{\Diamond}}
\newcommand{\varbox}{\mathrm{var}_{\Box}}

\newcommand{\ddiamondcp}[2]{{\langle\/ #1 \/\rangle}_{\!\mathit{cp}}^{#2}\,}

\usepackage{amsthm}

\newtheorem{lemma}{Lemma}
\newtheorem{definition}{Definition}
\newtheorem{theorem}{Theorem}
\newtheorem{corollary}{Corollary}

\theoremstyle{definition}
\newtheorem{example}{Example}
\newtheorem{remark}{Remark}

\begin{document}

\begin{frontmatter}



\title{Distributed Synthesis for\\ Parameterized Temporal Logics\tnoteref{thanks}}
\tnotetext[thanks]{Supported by the projects ASDPS (JA~2357/2--1) and TriCS (ZI~1516/1--1) of the German Research Foundation (DFG) and by the grant  OSARES (No.~683300) of the European Research Council (ERC).}


\author{Swen Jacobs}
\ead{jacobs@react.uni-saarland.de}

\author{Leander Tentrup}
\ead{tentrup@react.uni-saarland.de}

\author{Martin Zimmermann}
\ead{zimmermann@react.uni-saarland.de}

\address{Reactive Systems Group, Saarland University, 66123 Saarbr{\"u}cken, Germany}

\begin{abstract}
  We consider the synthesis of distributed implementations for specifications in parameterized temporal logics such as $\prompt$, which extends $\ltl$ by temporal operators equipped with parameters that bound their scope.
  For single process synthesis, it is well-established that such parametric extensions do not increase worst-case complexities.
  For synchronous distributed systems, we show that, despite being more powerful, the realizability problem for $\prompt$ is not harder than its $\ltl$ counterpart.
  For asynchronous systems, we have to express scheduling assumptions and therefore consider an assume-guarantee synthesis problem. As asynchronous distributed synthesis is already undecidable for $\ltl$, we give a semi-decision procedure for the $\prompt$ assume-guarantee synthesis problem based on bounded synthesis. Finally, we show that our results extend to the stronger logics $\pltl$ and $\pldl$. 
\end{abstract}

\end{frontmatter}

\bgroup
\renewcommand{\phi}{\varphi}

\section{Introduction}

Linear Temporal Logic~\cite{Pnueli77} ($\ltl$) is the most prominent specification language for reactive systems and the basis for industrial languages like ForSpec~\cite{Forspec02} and PSL~\cite{EisnerFismanPSL}. Its advantages include a compact variable-free syntax and intuitive semantics as well as the exponential compilation property, which explains its attractive algorithmic properties: every $\ltl$ formula can be translated into an equivalent B\"uchi automaton of exponential size~\cite{VardiWolper94}. This yields a polynomial space model checking algorithm and a doubly-exponential time algorithm for solving two-player games. Such games solve the monolithic $\ltl$ synthesis problem: given a specification, construct a correct-by-design implementation.

However, $\ltl$ lacks the ability to express timing constraints. For example, the request-response property~$\G(\mathit{req} \rightarrow \F \mathit{resp})$ requires that every request~$\mathit{req}$ is eventually responded to by a $\mathit{resp}$. It is satisfied even if the waiting times between requests and responses diverge, i.e., it is impossible to require that requests are granted within a fixed, but arbitrary, amount of time. While it is possible to encode an a-priori fixed bound for an eventually into $\ltl$, this requires prior knowledge of the system's granularity and incurs a blow-up when translated to automata, and is thus considered impractical.

To overcome this shortcoming of $\ltl$, Alur et al.\ introduced parametric $\ltl$~($\pltl$)~\cite{journals/tocl/AlurETP01}, which extends $\ltl$ with parameterized operators of the form $\F_{\le x}$ and $\G_{\le y}$, where $x$ and $y$ are variables. The formula~$\G(req \rightarrow \F_{\le x}\, resp)$ expresses  that every request is answered within an arbitrary, but fixed, number of steps~$\alpha(x)$. Here, $\alpha$ is a variable valuation, a mapping of variables to natural numbers. Typically, one is interested in whether a $\pltl$ formula is satisfied with respect to some variable valuation, e.g., model checking a transition system~$\tsys$ against a $\pltl$ specification~$\varphi$ amounts to determining whether there is an $\alpha$ such that every trace of $\tsys$ satisfies $\varphi$ with respect to $\alpha$. Alur et al.~\cite{journals/tocl/AlurETP01} showed that the $\pltl$ model checking problem is $\pspace$-complete. Due to monotonicity of the parameterized operators, one can assume that all variables $y$ in parameterized always operators $\G_{\le y}$ are mapped to zero, as variable valuations are quantified existentially in the problem statements. Dually, again due to monotonicity, one can assume that all variables $x$ in parameterized eventually operators $\F_{\le x}$ are mapped to the same value, namely the maximum of the bounds. Thus, in many cases the parameterized always operators and different variables for parameterized eventually operators are not necessary. 

Motivated by this, Kupferman et al.\ introduced $\prompt$~\cite{journals/fmsd/KupfermanPV09}, which can be seen as the fragment of $\pltl$ without the parameterized always operator and with a single bound~$k$ for the parameterized eventually operators. They proved that $\prompt$ model checking is $\pspace$-complete and solving $\prompt$ games is $\twoexp$-complete, i.e., not harder than $\ltl$ games. While the results of Alur et al.\ rely on involved pumping arguments, the results of Kupferman et al.\ are all based on the so-called alternating color technique, which basically allows to reduce $\prompt$ to $\ltl$.

Intuitively, one introduces a new proposition that is thought of as coloring traces of a system. Then, one replaces each parameterized eventually operator~$\F_{\le x}\varphi $ by an $\ltl$ formula requiring $\varphi$ to hold within at most one color change. If the distance between color changes is bounded from above, then satisfaction of the rewritten formula implies the existence of a bound~$k$ for the bounded eventually operators such that the original formula is satisfied with respect to $k$. Dually, if the distance between color changes is bounded from below, then the other implication holds: the original $\prompt$ formula implies the rewritten $\ltl$ formula. 

When applying this equivalence, one has to specify how the truth values for the new atomic proposition coloring the traces are determined. In a game setting (in particular in synthesis), the player who aims to satisfy the $\prompt$ formula determines these truth values and is required to change colors infinitely often. Then, a finite-state strategy automatically ensures an upper bound on the distance between color changes.

Later, the result on $\prompt$ games was extended to $\pltl$ games~\cite{journals/tcs/Zimmermann13}, relying on the monotonicity properties explained above and an application of the alternating color technique.
These results show that adding parameters to $\ltl$ does not increase the asymptotic complexity of the model checking and the game-solving problem, which is still true for even more expressive logics like Parametric Linear Dynamic Logic ($\pldl$)~\cite{journals/iandc/FaymonvilleZ17} and $\pltl$ and $\pldl$ with costs~\cite{Zimmermann2016}. The former logic is an extension of $\pltl$ with the full expressiveness of the $\omega$-regular languages. The latter logics are evaluated in weighted systems and generalize $\pltl$ and $\pldl$ by bounding the parameterized operators in the accumulated weight instead of  bounding them in time. 

The synthesis problems mentioned above assume a setting of complete information, i.e., every part of the system has a complete view on the system as a whole. However, this setting is unrealistic in distributed systems.
Based on this observation, \emph{distributed synthesis} is defined as the problem of synthesizing multiple components with incomplete information.
Since there are specifications that are not implementable, one differentiates synthesis from the corresponding decision problem, i.e., the \emph{realizability} problem of a formal specification.
We focus on the latter, but note that typically algorithms for the realizability problem also solve the synthesis problem, as they rely on constructing implementations to prove realizability. This also holds in our work here.

The realizability problem for distributed systems dates back to work of Pnueli and Rosner in the early nineties~\cite{conf/focs/PnueliR90}.
They showed that the realizability problem for LTL becomes undecidable already for the simple architecture of two processes with pairwise different inputs.
In subsequent work, it was shown that certain classes of architectures, like pipelines and rings, can still be synthesized automatically~\cite{conf/lics/KupfermanV01, DBLP:conf/fsttcs/MohalikW03}.
Later, a complete characterization of the architectures for which the realizability problem is decidable was given by Finkbeiner and Schewe by the \emph{information fork} criterion~\cite{conf/lics/FinkbeinerS05}.
Intuitively, an architecture contains an information fork if there is an information flow from the environment to two different processes where the information to one process is hidden from the other and vice versa.
The distributed realizability problem is decidable exactly for those architectures without an information fork.
Beyond decidability results, semi-decision procedures like bounded synthesis~\cite{journals/sttt/FinkbeinerS13} give an architecture-independent synthesis method that is particularly well-suited for finding small-sized implementations. Bounded synthesis searches for finite-state implementations of a fixed size by encoding the problem as a constraint system in a decidable first-order theory. In case of a positive answer, the result is returned, otherwise the bound is increased. If there is an upper bound on the size of a finite-state implementation, then bounded synthesis is a complete decision procedure, as it can be stopped if the upper bound is reached without a positive answer. If there is no such upper bound, it is indeed a semi-decision procedure that finds an implementation if one exists, but runs forever otherwise. 

\subsection{Our Contributions} 
As mentioned above, one can add parameters to $\ltl$ for free: the complexity of the model checking problem and of solving infinite games does not increase. This raises the question whether this is also true for distributed realizability of parametric temporal logics. 
For synchronous systems, we can answer this question affirmatively. 
For every class of architectures with decidable $\ltl$ realizability, the $\prompt$ realizability problem is decidable, too.
To show this, we apply the alternating color technique~\cite{journals/fmsd/KupfermanPV09} to reduce the distributed realizability problem of $\prompt$ to the one of $\ltl$: one can again add parameterized operators to $\ltl$ for~free. To prove this result, we add a new process whose only task it is to determine a coloring with the fresh proposition. By ensuring that the new process does not introduce an information fork we obtain decidability for the same class of architectures as for $\ltl$.

For asynchronous systems, the environment is typically assumed to take over the responsibility for the scheduling decision~\cite{conf/lopstr/ScheweF06}.
Consequently, the resulting schedules may be unrealistic, e.g., one process may not be scheduled at all.
While \emph{fairness} assumptions such as ``every process is scheduled infinitely often'' solve this problem for $\ltl$ specifications, they are insufficient for $\prompt$: a fair scheduler can still delay process activations arbitrarily long and thereby prevent the system from satisfying its $\prompt$ specification for any bound~$k$. \emph{Bounded fair} scheduling, where every process is guaranteed to be scheduled in bounded intervals, overcomes this problem.
Since bounded fairness can be expressed in $\prompt$, the realizability problem in asynchronous architectures can be formulated more generally as an assume-guarantee realizability problem that consists of two $\prompt$ specifications. 
Hence, we have to deal with two colorings of the traces when applying the alternating color technique: One induces bounds on the parameterized eventually operators in the assumption, the other on the bounds on the parameterized eventually operators in the guarantee. 

We give a semi-decision procedure for this problem based on a new method for checking emptiness of two-colored B\"uchi graphs~\cite{journals/fmsd/KupfermanPV09} and an extension of bounded synthesis~\cite{journals/sttt/FinkbeinerS13}.
As asynchronous $\ltl$ realizability for architectures with more than one process is undecidable~\cite{conf/lopstr/ScheweF06}, the same result holds for $\prompt$ realizability. Thus, the semi-decision procedure is the best result one can hope for.
Decidability in the one process case, which holds for $\ltl$~\cite{conf/lopstr/ScheweF06}, is left open for $\prompt$.

Finally, we show that all these results also hold for $\pltl$ and $\pldl$, even stronger logics to which the alternating color technique and bounded synthesis are still applicable. 

This is a revised and extended version of a paper that appeared at \mbox{GandALF} 2016~\cite{JacobsTZ16}.

\subsection{Related Work}
There is a rich literature regarding the synthesis of distributed systems from global $\omega$-regular specifications~\cite{conf/focs/PnueliR90,conf/lics/KupfermanV01,DBLP:conf/fsttcs/MohalikW03,conf/lics/FinkbeinerS05,conf/fmcad/ChatterjeeHOP13,journals/ipl/Schewe14,conf/tacas/FinkbeinerT14,journals/corr/FinkbeinerT15}.
We are not aware of work that is concerned with the realizability of parameterized logics in this setting.
For local specifications, i.e., specifications that only relate the inputs and outputs of single processes, the realizability problem becomes decidable for a larger class of architectures~\cite{conf/icalp/MadhusudanT01}.
An extension of these results to context-free languages was given by Fridman and Puchala~\cite{journals/acta/FridmanP14}.
The realizability problem for asynchronous systems and LTL specifications is undecidable for architectures with more than one process to be synthesized~\cite{conf/lopstr/ScheweF06}.
Later, Gastin et al.\ showed decidability of a restricted specification language and certain types of architectures, i.e., well-connected~\cite{journals/fmsd/GastinSZ09} and acyclic~\cite{journals/tocl/GastinS13} ones.
Bounded synthesis~\cite{journals/sttt/FinkbeinerS13,conf/tacas/FaymonvilleFRT17} provides a flexible  synthesis framework that can be used in both the asynchronous and the synchronous setting, based on a semi-decision procedure.

\subsection{Structure} In Section~\ref{prompt}, we introduce $\prompt$ and the alternating color technique. In Section~\ref{sec:synchronous_distributed_synthesis}, we consider synchronous distributed synthesis for $\prompt$ and the asynchronous case in Section~\ref{sec:asynchronous_distributed_synthesis}. Then, in Section~\ref{sec:beyond_prompt_ltl}, we consider both problems for the more expressive logics~$\pltl$ and $\pldl$. We conclude in Section~\ref{sec:conc} with a discussion of future work.

\section{$\prompt$}\label{prompt}

Throughout this work, we fix a set~$\ap$ of atomic propositions. The formulas of $\prompt$
are given by the grammar
\begin{equation*}\phi \Coloneqq a \mid \neg a \mid \phi \wedge \phi \mid \phi \vee
\phi
  \mid \X \phi \mid \phi \U \phi \mid \phi \R \phi \mid
  \Fp \phi
  \enspace,\end{equation*}
where $a \in \ap$ is an atomic proposition, $\neg,\wedge,\vee$ are the usual Boolean operators, and $\X, \U, \R$ are the $\ltl$ operators next, until, and release. 
We use the derived operators
$\true \coloneqq a \vee \neg a$ and $\false \coloneqq a \wedge \neg a$ for some fixed $a \in \ap$,
and $\F \phi \coloneqq \true \U \phi$ and $\G \phi \coloneqq \false \R \phi$ as usual.
Furthermore, we use $\phi \implies \psi$ as a shorthand for $\neg \phi \vee \psi$
if the antecedent~$\phi$ is an $\Fp$-free formula (since in that case we can transform $\neg \varphi$ into negation normal form in the fragment above).
We define the size of $\phi$ to be the number of sub-fomulas of $\phi$.
The satisfaction relation for $\prompt$ is defined between an $\omega$-word~$w = w_0 w_1 w_2 \cdots \in \left( \pow{\ap} \right)^{ \omega }$, a
position~$n \in \nats$, a bound~$k$ for the prompt-eventually operators, and a~$\prompt$ formula. 
\begin{itemize}

	\item $(w,n,k)\models a$ if, and only if, $a \in w_n$.
	
	\item $(w,n,k)\models \neg a$ if, and only if, $a \notin w_n$.
	
	\item $(w,n,k)\models \phi_0\wedge \phi_1$ if, and only if, $(w,n,k)\models \phi_0$ and $(w,n,k)\models \phi_1$.
	
	\item $(w,n,k)\models \phi_0\vee \phi_1$ if, and only if, $(w,n,k)\models \phi_0$ or $(w,n,k)\models \phi_1$.
	
	\item $(w,n,k)\models \X\phi$ if, and only if, $(w,n+1,k)\models \phi$.
	
	\item $(w,n,k)\models \phi_0\U\phi_1$ if, and only if, there exists a $j \ge 0$ such that $(w,n+j,k)\models\phi_1$ and $(w,n+j',k)\models\phi_0$ for every $j'$ in the range~$0 \le j' < j$.
	
	\item $(w,n,k)\models \phi_0\R\phi_1$ if, and only if, for all $j \ge 0$: $(w,n+j,k)\models\phi_1$ or $(w,n+j',k)\models\phi_0$ for some $j'$ in the range~$0 \le j' < j$.

	\item $(w,n,k)\models\Fp\phi$ if, and only if, there exists a $j$ in the range
$0\le j \le k$ such that $(w,n+j,k)\models\phi$.

\end{itemize}

For the sake of brevity, we write $(w,k) \models \phi$ instead of
$(w,0,k) \models \phi$ and say that $w$ is a model of $\phi$ with
respect to $k$. Note that $(w, n, k) \models \phi$ implies $(w, n, k') \models \phi$ for every $k' \ge k$, i.e., satisfaction with respect to $k$ is an upward-closed property.

\paragraph[The Alternating Color Technique]{The Alternating Color Technique} \label{subsection_altcolor}

In this subsection, we recall the alternating color technique, which 
Kupferman et al.\ introduced to solve model checking, assume-guarantee model checking, and the
realizability problem for $\prompt$ specifications~\cite{journals/fmsd/KupfermanPV09}.

Let $r\notin \ap$ be a fixed fresh proposition. An
$\omega$-word~$w'\in\left(2^{\ap\cup\{r\}}\right)^{\omega}$ is an $r$-coloring of
$w\in\left(2^{\ap}\right)^{\omega}$ if $w_n'\cap \ap=w_n$, i.e., $w_n$ and $w_n'$
coincide on all propositions in $\ap$. The additional proposition~$r$ can be
thought of as the color of $w_n'$: we say
that \emph{the color changes} at position~$n$, if $n=0$ or if the truth values of $r$ in $w_{n-1}'$
and in $w_n'$ are not equal. In this situation, we say that $n$ is a \emph{change point}.
An \emph{$r$-block} is a maximal infix~$w_m' \cdots w_{n}'$ of $w'$ such that the color changes at $m$ and $n+1$, but not in between. 

Let $k \ge 1$. We say
that $w'$ is \emph{$k$-spaced} if the color changes infinitely often and each
$r$-block has length at least $k$, and we say that $w'$ is \emph{$k$-bounded}, if each
$r$-block has length at most $k$. Note that $k$-boundedness implies that the color changes
infinitely often.

Given a $\prompt$ formula~$\phi$, let $\rel_r(\phi)$ denote the formula obtained by inductively replacing every sub-formula~$\Fp\psi$ by
\begin{equation*}
(r\implies (r\U(\neg r\U \rel_r(\psi))))\wedge(\neg r\implies (\neg r\U(
r\U \rel_r(\psi))))\enspace,
\end{equation*}
which is only linearly larger than $\phi$ and requires every prompt eventually to be satisfied within at most one color change (not counting the position where $\psi$ holds). 
Furthermore, the formula
$\alt_r = \GF r\wedge \GF\neg r$ is satisfied if the colors change infinitely
often. Finally, we define the $\ltl$ formula~$c_r(\phi) = \rel_r(\phi) \wedge \alt_r$.
Kupferman et al.\ showed that $\phi$ and $c_r(\phi)$ are in some sense equivalent on $\omega$-words
which are bounded and spaced. 

\begin{lemma}[Lemma~2.1 of \cite{journals/fmsd/KupfermanPV09}]
\label{lemma_alternatingcolor}
Let $\phi$ be a $\prompt$ formula, and let $w \in \left( \pow{\ap} \right)^{ \omega }$.
\begin{enumerate}
  \item \label{lemma_alternatingcolor_pltltoltl}
If $(w,k)\models \phi$, then $w' \models c_r(\phi)$ for every $k$-spaced $r$-coloring~$w'$ of~$w$.
  
  \item \label{lemma_alternatingcolor_ltltopltl}
If $w'$ is a $k$-bounded $r$-coloring of $w$ with $w' \models c_r(\phi)$, then $(w,2k)\models\phi$.
\end{enumerate}
\end{lemma}
Whenever possible, we drop the subscript~$r$ for the sake of readability, if $r$ is clear from context. However, when we consider asynchronous systems in Section~\ref{sec:asynchronous_distributed_synthesis}, we need to relativize two formulas with different colors, which necessitates the introduction of the subscripts.

\section{Synchronous Distributed Synthesis} \label{sec:synchronous_distributed_synthesis}

$\prompt$ specifications can give guarantees that $\ltl$ cannot, for example by asserting not only that requests to a system are answered \emph{eventually}, but also that there is an \emph{upper bound} on the reaction time.
This is especially important in distributed systems, since such timing constraints become more difficult to implement because of information flows between the various parts of the system.

Consider, for example, a distributed computation system, where
a central server gets \emph{important} and \emph{unimportant} tasks, and can forward tasks to a number of clients. 
A client can either enqueue the task, which means that it will be processed \emph{eventually}, or clear the client-side queue and process the task immediately. The latter operation is very costly (we have to remember the open tasks as they still need to be completed), but guarantees an upper bound on the completion time.
While in $\ltl$ we can only specify that all incoming tasks are processed eventually, in $\prompt$ we can specify that the answer time to important tasks is bounded by the formula
$\G (\mathit{important\text{-}task} \rightarrow \Fp \mathit{finished\text{-}task})$.\footnote{A similar constraint could be simulated in $\ltl$ by writing that on every important incoming task, the worker queues are cleared. This, however, removes implementation freedom and requires the developer to determine how to implement the feature, instead of letting the synthesis algorithm decide.}

Let us now formalize the distributed realizability problem.
Let $X$ and $Y$ be finite and pairwise disjoint sets of variables.
A \emph{valuation} of $X$ is a subset of $X$; thus, the set of all valuations of $X$ is $\pow{X}$.
For $w = w_0 w_1 w_2 \cdots \in (\pow{X})^\omega$ and $w'=w'_0 w'_1 w'_2 \cdots \in (\pow{Y})^\omega$, let $w \cup w' = (w_0 \cup w'_0) (w_1 \cup w'_1) (w_2 \cup w'_2) \cdots \in (\pow{X \cup Y})^\omega$. 

\paragraph{Strategies}

A \emph{strategy} $f \colon \strat{X}{Y}$ maps a history of valuations of $X$ to a valuation of $Y$.
The behavior of a strategy $f\colon \strat{X}{Y}$ is characterized by an infinite tree that branches by the valuations of $X$ and whose nodes $w \in (\pow{X})^*$ are labeled with the strategic choice $f(w)$.
For an infinite word $w = w_0 w_1 w_2 \cdots \in (\pow{X})^\omega$, the corresponding labeled path is defined as $(f(\epsilon) \cup w_0)(f(w_0) \cup w_1)(f(w_0 w_1) \cup w_2)\cdots \in (2^{X \cup Y})^\omega$.
We lift the set containment operator $\in$ to the containment of a labeled path $w = w_0 w_1 w_2 \cdots \in (2^{X \cup Y})^\omega$ in a strategy tree induced by $f \colon \strat{X}{Y}$, i.e., $w \in f$ if, and only if, $f(\epsilon) = w_0 \cap Y$ and $f((w_0 \cap X) \cdots (w_i \cap X)) = w_{i+1} \cap Y$ for all $i \geq 0$.

A \emph{$\pow{Y}$-labeled $\pow{X}$-transition system} $\tsys$ is a tuple $\tuple{S,s_0,\Delta,l}$ where
$S$ is a finite set of states,
$s_0 \in S$ is the designated initial state,
$\Delta \colon \fun{S \times \pow{X}}{S}$ is the transition function, and
$l \colon \fun{S}{\pow{Y}}$ is the state-labeling.
We generalize the transition function to sequences over $\pow{X}$ by defining $\Delta^* \colon \fun{(\pow{X})^*}{S}$ recursively as $\Delta^*(\epsilon) = s_0$ and $\Delta^*(w_0 \cdots w_{n-1} w_n) = \Delta(\Delta^*(w_0 \cdots w_{n-1}), w_n)$ for $w_0 \cdots w_{n-1} w_n \in (\pow{X})^+$.
A transition system $\tsys$ \emph{generates} the strategy $f$ if $f(w) = l(\Delta^*(w))$ for every $w \in (\pow{X})^*$.
A strategy $f$ is called \emph{finite-state} if there exists a transition system that generates $f$.

To reason about distributed systems, we have to combine strategies with different inputs, which we call the distributed product. To this end, we have to introduce widenings of strategies, which intuitively enlarge their domains with new atomic propositions that are ignored. Also, we need projections, which remove outputs from strategies.

In the following, we formally introduce these concepts.
A visualization is given in Fig.~\ref{fig:distributed_product}.
\begin{definition}[Distributed Product]
  Let $X, X', Y$, and $Y'$ be finite sets such that $Y$ and $Y'$ are disjoint.
  Further, let $f\colon \strat{X}{Y}$ and $f'\colon \strat{X}{Y'}$ be two strategies with the same domain but different co-domains $\pow{Y}$ and $\pow{Y'}$.
  \begin{itemize}
    \item The \emph{product} $f \times f' \colon \strat{X}{Y \cup Y'}$ of $f$ and $f'$ is defined as $(f \times f')(w) = f(w) \cup f'(w)$ for every $w \in (\pow{X})^*$.
    \item The $\pow{X}$-projection of a sequence $w_0 \cdots w_n \in (2^{X \cup X'})^*$ is $\proj_{\pow{X}}(w_0 \cdots w_n) = (w_0 \cap X) \cdots (w_n \cap X) \in (\pow{X})^*$. 
    \item The $\pow{X'}$-widening of $f $ is defined as $\wide_{\pow{X'}}(f) \colon \strat{X \cup X'}{Y}$ with $\wide_{\pow{X'}}(f)(w) = f(\proj_{\pow{X}}(w))$ for $w \in (2^{X \cup X'})^*$. 
    \item Given some $g\colon \strat{X'}{Y'}$, the \emph{distributed product} $f \distprod g \colon \strat{X \cup X'}{Y \cup Y'}$ is defined as the product $\wide_{2^{X' \setminus X}}(f) \times \wide_{2^{X \setminus X'}}(g)$.
  \end{itemize} 
\end{definition}
Analogously, for transition systems $\tsys = \tuple{S,s_0,\Delta,l}$ and $\tsys' = \tuple{S',s_0',\Delta',l'}$ the distributed product, written $\tsys \distprod \tsys'$, is defined as the transition system $\tuple{S \times S', (s_0,s_0'), \Delta^\distprod, l^\distprod}$, where $\Delta^\distprod((s,s'),w) = (s'',s''')$ if, and only if, $\Delta(s,w) = s''$ and $\Delta'(s',w) = s'''$, and $l^\distprod(s,s') = l(s) \cup l'(s')$.
\begin{remark}
  The strategy generated by $\tsys \distprod \tsys'$ is equal to the distributed product of the strategies generated by $\tsys$ and $\tsys'$.
\end{remark}

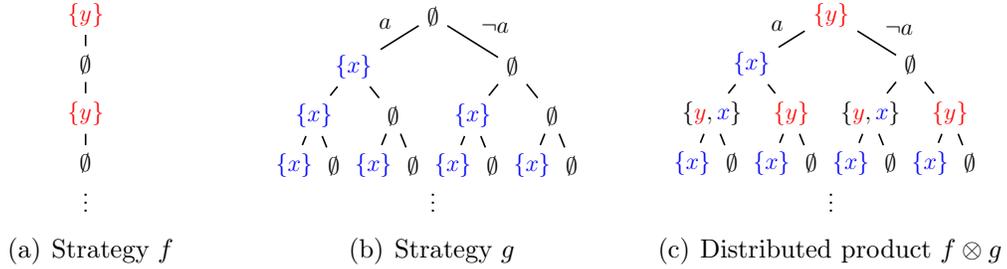
\begin{figure}[t]
  \centering
  \subfigure[Strategy $f$]{
      \begin{tikzpicture}[auto,node distance=1cm,semithick,scale=0.75,transform shape]
      
  \tikzstyle{treenode} = [fill,circle,inner sep=0,minimum size=5pt]
  
  \node[red](root) {$\set{y}$} [level distance=25pt]
    child
    {
      node[] {$\emptyset$}
      child
      {
        node[red] {$\set{y}$}
        child
        {
          node[] {$\emptyset$}
        }
      }
    }
    ;
    
  \node[left=30pt of root] {};
  \node[right=30pt of root] {};
  
  \node[below=70pt of root] {$\vdots$};
  
  \end{tikzpicture}
    \label{fig:example_strategy_linear}
  }\hfill
  \subfigure[Strategy $g$]{
      \begin{tikzpicture}[auto,node distance=1cm,semithick,scale=0.75,transform shape]
      
  \tikzstyle{treenode} = [fill,circle,inner sep=0,minimum size=5pt]
  \tikzstyle{treenodenfill} = [draw,circle,inner sep=0,minimum size=5pt]
  \tikzstyle{usepath} = [line width=1.2pt]
  
  \node[](root) {$\emptyset$} [level distance=25pt,
    level 1/.style={sibling distance=80pt},
    level 2/.style={sibling distance=40pt},
    level 3/.style={sibling distance=20pt}]
  
    child
    {
      node[blue] {$\set{x}$} 
      child
      {
        node[blue] {$\set{x}$} 
        child {node[blue] {$\set{x}$}}
        child {node[] {$\emptyset$}}
      }
      child
      {
        node[] {$\emptyset$}
        child {node[blue] {$\set{x}$}}
        child {node[] {$\emptyset$}}
      }
      edge from parent node[above left] {$a$}
    }
    child
    {
      node[] {$\emptyset$}
      child
      {
        node[blue] {$\set{x}$}
        child {node[blue] {$\set{x}$}}
        child {node[] {$\emptyset$}}
      }
      child
      {
        node[] {$\emptyset$}
        child {node[blue] {$\set{x}$}}
        child {node[] {$\emptyset$}}
      }
      edge from parent node[above right] {$\neg a$}
    }
    ;
  
  \node[below=70pt of root] {$\vdots$};
  
  \end{tikzpicture}
    \label{fig:example_strategy_branching}
  }\hfill
  \subfigure[Distributed product $f \otimes g$]{
      \begin{tikzpicture}[auto,node distance=1cm,semithick,scale=0.75,transform shape]
      
  \tikzstyle{treenode} = [fill,circle,inner sep=0,minimum size=5pt]
  \tikzstyle{treenodenfill} = [draw,circle,inner sep=0,minimum size=5pt]
  \tikzstyle{usepath} = [line width=1.2pt]
  
  \node[red](root) {$\set{y}$} [level distance=25pt,
    level 1/.style={sibling distance=80pt},
    level 2/.style={sibling distance=40pt},
    level 3/.style={sibling distance=20pt}]
  
    child
    {
      node[blue] {$\set{x}$} 
      child
      {
        node[] {$\set{{\color{red} y},{\color{blue} x}}$}
        child {node[blue] {$\set{x}$}}
        child {node[] {$\emptyset$}}
      }
      child
      {
        node[red] {$\set{y}$}
        child {node[blue] {$\set{x}$}}
        child {node[] {$\emptyset$}}
      }
      edge from parent node[above left] {$a$}
    }
    child
    {
      node[] {$\emptyset$}
      child
      {
        node[] {$\set{{\color{red} y},{\color{blue} x}}$}
        child {node[blue] {$\set{x}$}}
        child {node[] {$\emptyset$}}
      }
      child
      {
        node[red] {$\set{y}$}
        child {node[blue] {$\set{x}$}}
        child {node[] {$\emptyset$}}
      }
      edge from parent node[above right] {$\neg a$}
    }
    ;
  
  \node[below=70pt of root] {$\vdots$};
  
  \end{tikzpicture}
    \label{fig:example_distributed_product}
  }
  \caption{Visualization of strategies $f \colon \strat{\emptyset}{\set{y}}$ and $g \colon \strat{\set{a}}{\set{x}}$ as infinite trees is shown in \subref{fig:example_strategy_linear} and \subref{fig:example_strategy_branching}, respectively. The distributed product $f \otimes g$ is equal to the product of $g$ and the $2^\set{a}$-widening of $f$ and is depicted in \subref{fig:example_distributed_product}.}
  \label{fig:distributed_product}
\end{figure}

We define the satisfaction of a $\prompt$ formula $\varphi$ (over propositions $X \cup Y$) on strategy~$f$ with respect to the bound~$k$, written $(f,k) \models \phi$ for short, as $(w,k) \models \varphi$ for all paths $w \in f$.

\paragraph{Distributed Systems}

We characterize a distributed system as a set of processes with a fixed communication topology, called an \emph{architecture} in the following.
Recall that $\ap$ is the set of atomic propositions used to build formulas.
An \emph{architecture} $\arch$ is a tuple $\tuple{P,\penv,\{I_p\}_{p \in P}, \{O_p\}_{p \in P}}$, where $P$ is the finite set of processes and $p_\mathit{env} \in P$ is the distinct environment process. We denote by $\pminus = P \setminus \set{\penv}$ the set of system 
processes.

Given a process $p \in P$, the input and output signals of this process are $I_p \subseteq \ap$ and $O_p \subseteq \ap$, respectively, where we assume~$I_{\penv} = \emptyset$.
For $P' \subseteq P$, let $I_{P'} = \bigcup_{p \in P'} I_p$ and $O_{P'} = \bigcup_{p \in P'} O_p$. While processes may share the same inputs (in case of broadcasting), the outputs of processes must be pairwise disjoint, i.e., for all $p \neq p' \in P$ it holds that $O_p \cap O_{p'} = \emptyset$. Finally, we require that every input of a process originates from some other process, i.e., $I_P \subseteq O_P$.

An \emph{implementation} of a process $p \in \pminus$ is a strategy $f_p \colon \strat{I_p}{O_p}$ mapping finite input sequences to a valuation of the output variables.

\begin{example}

Figure~\ref{fig:architectures} shows example architectures $\arch_1$ and $\arch_2$:
\begin{itemize}
	\item $\arch_1 =  \langle  \set{\penv, p_1, p_2},\penv, \set{ I_\penv, I_{p_1}, I_{p_2}}, \set{ O_\penv, O_{p_1}, O_{p_2}}\rangle $ with
\begin{itemize}
	\item $I_\penv = \emptyset, I_{p_1} = \set{a}, I_{p_2} = \set{b} $ and
	\item $O_\penv = \set{a, b}, O_{p_1} = \set{c}, O_{p_2} = \set{d}$.
\end{itemize}
	\item $ \arch_2 = \langle 
    \set{\penv, p_1, p_2},
    \penv,
    \set{ I_\penv, I_{p_1}, I_{p_2} },
    \set{ O_\penv  , O_{p_1} , O_{p_2} }
  \rangle$ with 
  
  \begin{itemize}
  \item $ I_\penv = \emptyset, I_{p_1} = \set{a}, I_{p_2} = \set{b}$ and 
  \item $O_\penv = \set{a}, O_{p_1} = \set{b}, O_{p_2} = \set{c}$.
  \end{itemize}

\end{itemize}
The architecture $\arch_1$ in Fig.~\ref{fig:independent_architecture} contains two system processes, $p_1$ and $p_2$, and the environment process $\penv$.
The processes $p_1$ and $p_2$ receive the inputs~$a$ and $b$, respectively, from the environment and output $c$ and $d$, respectively.
Hence, the environment can provide process $p_1$ with information that is hidden from $p_2$ and vice versa.
In contrast, architecture $\arch_2$, depicted in Fig.~\ref{fig:pipeline_architecture}, is a pipeline architecture where information from the environment can only propagate through the pipeline processes $p_1$ and $p_2$.
\end{example}
\begin{figure}[t]
\centering
\subfigure[$\arch_1$]{
  \begin{tikzpicture}[->,>=stealth',shorten >=1pt,auto,node distance=1cm,semithick,scale=0.8,transform shape]
  
  \tikzstyle{every state}=[shape=rectangle]
  \tikzstyle{envstate}=[shape=circle,scale=0.95]
  
  \node[state,envstate]           (env) {$p_\mathit{env}$};
  \node[state,above right=0 and 1 of env] (P0)  {$p_1$};
  \node[state,below right=0 and 1 of env] (P1)  {$p_2$};
  
  \path (P0.east) edge node  {$c$} +(right:0.75)
        (P1.east) edge node  {$d$} +(right:0.75)
        (env) edge node [pos=0.65]{$a$} (P0)
        (env) edge node [pos=0.2]{$b$} (P1)
        ;


\end{tikzpicture}
  \label{fig:independent_architecture}
}\qquad\qquad\qquad%
\subfigure[$\arch_2$]{
  \begin{tikzpicture}[->,>=stealth',shorten >=1pt,auto,node distance=1cm,semithick,scale=0.8,transform shape]
  
  \tikzstyle{every state}=[shape=rectangle]
  \tikzstyle{envstate}=[shape=circle,scale=0.95]
  
  \node[state,envstate] (env)      {$p_\mathit{env}$};
  \node[state] (P0) [right=0.85 of env] {$p_1$};
  \node[state] (P1) [right=of P0]  {$p_2$};
  
  \path (env) edge node {$a$} (P0)
        (P0) edge node {$b$} (P1)
        (P1.east) edge node {$c$} +(right:0.75)
        ;
  
  \path[white] (P0.south) edge node {} +(down:0.75);

\end{tikzpicture}
  \label{fig:pipeline_architecture}
}
\caption[]{Examples for distributed architectures.}
\label{fig:architectures}
\end{figure}
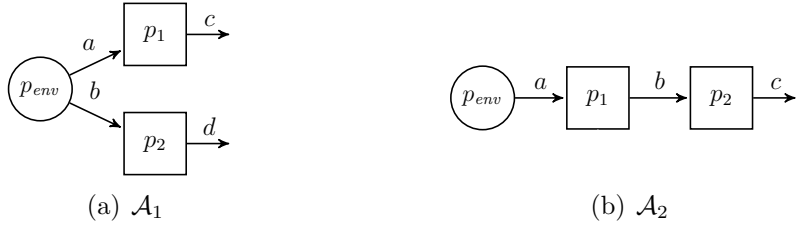

\paragraph{Distributed Realizability}

Let $\arch = \tuple{P,\penv,\{I_p\}_{p \in P}, \{O_p\}_{p \in P}}$ be an architecture.
The \emph{synchronous $\prompt$ realizability problem for $\arch$} is to decide, given a $\prompt$ formula $\varphi$, whether there exist a bound $k$ and a finite-state implementation $f_p$ for each process $p \in P^-$, such that the distributed product $\Distprod_{p \in P^-} f_p$ satisfies $\varphi$ with respect to $k$, i.e., $(\Distprod_{p \in P^-} f_p,k) \models \varphi$. In this case, we say that $\phi$ is realizable in $\arch$.
The synchronous $\ltl$ realizability problem is a special case of it, as $\ltl$ is a fragment of $\prompt$.

In the following, we show how to solve the synchronous $\prompt$ realizability problem.
In our reduction to synchronous $\ltl$ realizability, we introduce a new process that produces a coloring sequence needed for applying the alternating color technique~\cite{journals/fmsd/KupfermanPV09}.
Let $r \notin \ap$ be the fresh proposition introduced for the alternating color technique to relativize formulas and let $\arch = \tuple{P,\penv,\{I_p\}_{p \in P}, \{O_p\}_{p \in P}}$ be an architecture as above. We define the architecture $\arch^r$ as 
\[\tuple{P \cup \set{p_r},\penv,\{I_p\}_{p \in P} \cup \set{I_r},\{O_p\}_{p \in P} \cup \set{O_r}},\] where $I_r = \emptyset$ and $O_r = \set{r}$.
Intuitively, this describes an architecture where one additional process $p_r$ is responsible for providing sequences in $(\pow{\set{r}})^\omega$, i.e., a coloring by $r$. We show that $\varphi$ in $\arch$ and $c_r(\varphi)$  in $\arch^r$ are equi-realizable by applying the alternating color technique. As the processes are synchronized, the proof is similar to the one for the single-process case by Kupferman et al.~\cite{journals/fmsd/KupfermanPV09}.

\begin{theorem}
\label{thm_synchr_prompt2ltl}
  A $\prompt$ formula $\varphi$ is realizable in $\arch$ if, and only if, $c_r(\varphi)$ is realizable in $\arch^r$.
\end{theorem}
\begin{proof}%
  Let $\arch = \tuple{P,\penv,\{I_p\}_{p \in P}, \{O_p\}_{p \in P}}$ be an architecture and $\varphi$ be a $\prompt$ formula.
  
  Assume that the $\prompt$ formula $\varphi$ is realizable in $\arch$.
  Then, there exist finite-state strategies $f_p$ for $p \in P^-$ and a bound $k$ satisfying the synchronous $\prompt$ realizability problem $\tuple{\arch,\varphi}$.
  For every $w \in \Distprod_{p \in P^-} f_p$, it holds that $(w,k) \models \varphi$.
  By item~\ref{lemma_alternatingcolor_pltltoltl} of Lemma~\ref{lemma_alternatingcolor}, it holds that every $k$-spaced $r$-coloring $w'$ of $w$ satisfies $c_r(\varphi)$.
  Let $f_r \colon (2^\emptyset)^* \rightarrow 2^\set{r}$ be a (finite-state) strategy that produces the $k$-spaced sequence $(\emptyset^k \set{r}^k)^\omega$.
  Then, the process implementations $\set{f_p}_{p \in P^-}$ together with $f_r$ are a solution to the synchronous $\ltl$  realizability problem $\tuple{\arch^r,c_r(\varphi)}$.

  Now, assume that the $\ltl$ formula $c_r(\varphi)$ is realizable in the architecture~$\arch^r$.
  Thus, there exist finite-state strategies $f_p$ for $p \in P^-$ and a finite-state strategy $f_r$ for process~$p_r$. Note that the strategy~$f_r \colon \strat{\emptyset}{\set{r}}$ has a unique output~$w_r \in (\pow{{\set{r}}})^\omega$, as it has no inputs. We claim that $w_r$ is $k$-bounded, where $k$ is the number of states of the transition system~$\tsys = \tuple{S,s_0,\Delta,l}$ generating $f_r$. To see this, note that $f_r$ has no inputs, i.e., each state of $\tsys$ has a unique successor in $\Delta$, and the unique run of $\tsys$ on $\emptyset^\omega$ ends up in a loop which is traversed ad infinitum. As the output~$w_r$ has infinitely many change points (since $c_r(\varphi)$ is realizable in $\arch^r$), the loop contains at least one state~$s$ labeled by $l(s) = \emptyset$ and at last one state~$s'$ with $l(s') = \set{r}$. Thus, the maximal length of a block of $w_r$ is bounded by the length of the loop, which in turn is bounded by the size of $\tsys$. 

  Hence, for every $w \in \Distprod_{p \in P^-} f_p$, the word $w_r \cup w$ is a $k$-bounded $r$-coloring of $w$ with $w_r \cup w \models \rel_r(\varphi)$.
  By item~\ref{lemma_alternatingcolor_ltltopltl} of Lemma~\ref{lemma_alternatingcolor}, for all such $w$ it holds that $(w,2k) \models \varphi$.
  Hence, $\set{f_p}_{p \in P^-}$ together with the bound~$2k$ is a solution to the synchronous $\prompt$ realizability problem.
\end{proof}

Theorem~\ref{thm_synchr_prompt2ltl} allows us to reduce the distributed realizability problem of $\prompt$ to the distributed realizability problem of $\ltl$ in a strategy-preserving manner as shown in the accompanying proof.
In particular, we are able to reuse semi-decision procedures for the latter, such as bounded synthesis~\cite{journals/sttt/FinkbeinerS13}, to effectively construct small solutions.

To conclude, we show that the newly introduced process $p_r$ also preserves the property whether the architecture has an \emph{information fork}~\cite{conf/lics/FinkbeinerS05}.
Formally, consider tuples $\tuple{P', V', p, p'}$, where $P'$ is a subset of the processes, $V'$ is a subset of the variables disjoint from $I_p \cup I_{p'}$, and $p,p' \in  P^- \setminus P'$ are two different processes.
Such a tuple is an information fork in $\arch$ if $P'$ together with the edges that are labeled with at least one variable from $V'$ forms a sub-graph of $\arch$ rooted in the environment and there exist two nodes $q, q' \in P'$ that have edges to $p, p'$, respectively, such that $O_\set{q,p} \nsubseteq I_{p'}$ and $O_\set{q',p'} \nsubseteq I_p$.
For example, the architecture in Fig.~\ref{fig:independent_architecture} contains the information fork $(\set{\penv}, \emptyset, p_1, p_2)$, while the pipeline architecture depicted in Fig.~\ref{fig:pipeline_architecture} has no information forks.

\begin{lemma}
  $\arch^r$ contains an information fork if, and only if, $\arch$ contains an information fork.
\end{lemma}
\begin{proof}
  The \emph{if} direction follows immediately by construction: if $\tuple{P', V', p, p'}$ is an information fork in $\arch$, then it is an information fork in $\arch^r$ as well.
  Hence, assume $\tuple{P', V', p, p'}$ is an information fork in $\arch^r$.
  It holds that neither $p_r = p$ nor $p_r = p'$ since $p_r$ has no incoming edges.
  As $I_{p_r} = \emptyset$, $p_r$ cannot be in a sub-graph that is rooted in the environment, hence, $p_r \notin P'$ and $r \notin V'$.
  It follows that $\tuple{P', V', p, p'}$ is an information fork in $\arch$.
\end{proof}

Thus, we can use well-known results for the decidability of distributed realizability for $\ltl$ and weakly ordered architectures~\cite{conf/lics/FinkbeinerS05}, i.e., those without an information fork.

\begin{corollary}
\label{coro_promptsynthesis}
  Let $\arch$ be an architecture.
  The synchronous $\prompt$ realizability problem for $\arch$ is decidable if, and only if, $\arch$ is weakly ordered.
\end{corollary}

\section{Asynchronous Distributed Synthesis} \label{sec:asynchronous_distributed_synthesis}

The asynchronous system model is a generalization of the synchronous model discussed in the previous section.
In an asynchronous system, not all processes are scheduled at the same time.
We model the scheduler as part of the environment, i.e., at any given time the environment additionally signals whether a process is enabled.
The resulting distributed realizability problem is already undecidable for $\ltl$ specifications and systems with more than one process~\cite{conf/lopstr/ScheweF06}.

We have to adapt the definition of the synchronous $\prompt$ realizability problem for the asynchronous setting.
Using the definition from Section~\ref{sec:synchronous_distributed_synthesis}, the system can never satisfy a $\prompt$ formula if the scheduler is part of the environment, since it may delay scheduling indefinitely. Moreover, even if the scheduler is assumed to be fair, it can still build increasing delay blocks between process activation times such that it is impossible for the system to guarantee any bound $k \in \nats$.
Hence, we employ the concept of \emph{bounded fair} schedulers and allow the system bound to depend on the scheduler bound.
More generally, this is a typical instance of an assume-guarantee specification: under the assumption that the scheduler is bounded fair, the system satisfies its specification.
In the following, we formally introduce the distributed realizability problem for asynchronous systems and assume-guarantee specifications.

\paragraph{Scheduling}

To model scheduling, we introduce an additional set $\Sched = \set{\sched_p \mid p \in \pminus}$ of atomic propositions. The valuation of $\sched_p$ indicates whether system process~$p$ is currently scheduled or not.
Given a (synchronous) architecture $\arch = \tuple{P,\penv,\{I_p\}_{p \in P}, \{O_p\}_{p \in P}}$, we define the asynchronous architecture $\arch^*$ as the architecture with the environment output $O^*_{\penv} = O_{\penv} \cup \Sched$.
Furthermore, we extend the input $I_p$ of a process by its scheduling variable $\sched_p$, i.e., $I_p^* = I_p \cup \set{\sched_p}$ for each $p \in P^-$.
The environment can decide at every step which processes to schedule.
When a process is not scheduled, its \emph{state}---and thereby its outputs---do not change~\cite{journals/sttt/FinkbeinerS13}.
Formally, for $p \in P^-$, let $f_p$ be a finite-state strategy for a process~$p$ and $\tsys_p = \tuple{S,s_0,\Delta,l}$ a transition system that generates~$f_p$.
For every path $w = w_0 w_1 w_2 \cdots \in (\pow{I_p^*})^\omega$, it holds that if $\sched_p \notin w_i$ for some $i \in \nats$, then $\Delta^*(w[i]) = \Delta^*(w[i+1])$, where $w[i]$ denotes the prefix~$w_0 w_1 \cdots w_{i}$ of $w$.
For the remainder of this section, we will only consider such strategies.


\begin{definition}[Assume-Guarantee Realizability]
  A $\prompt$ assume-guarantee specification~$\tuple{\varphi,\psi}$ consists of a pair of $\prompt$ formulas.
  The asynchronous $\prompt$ assume-guarantee realizability problem asks, given an asynchronous architecture~$\arch^*$ and $\tuple{\varphi,\psi}$ as above, whether for each process $p \in P^-$ there exists a finite-state strategy~$f_p$ such that for every bound~$k$ on the assumption there is a bound $l$ on the guarantee such that for every $w \in \Distprod_{p \in P^-} f_p$, we have that $(w, k) \models \varphi$ implies $(w, l) \models \psi$.
  In this case, we say that $\Distprod_{p \in P^-} f_p$ satisfies $\tuple{\varphi,\psi}$.
\end{definition}

Consider the bounded fairness specification discussed above, which is expressed by the formula~$\varphi = \bigwedge_{p \in \pminus} \GFp \sched_p$, i.e., for every point in time, every $p$ is scheduled within a bounded number of steps.
We use $\varphi$ as an assumption on the environment which implies that the guarantee $\psi$ only has to be satisfied if $\varphi$ holds.
Consider, for example, the asynchronous architecture corresponding to Fig.~\ref{fig:independent_architecture} and the $\prompt$ specification $\psi = \G (\Fp c \land \Fp \neg c \land \Fp d \land \Fp \neg d)$.
Even when we assume a fair scheduler, i.e., $\varphi = \GF \sched_{p_1} \land \GF \sched_{p_2}$, the environment can prevent one process from satisfying the specification for any bound $l$.
This problem is fixed by assuming the scheduler to be bounded fair, i.e., $\varphi = \GFp \sched_{p_1} \land \GFp \sched_{p_2}$.
Then, there exist realizing implementations for processes $p_1$ and $p_2$ (that alternate between enabling and disabling the output), and the bound on the guarantee is $l = 2 \cdot k$ for every bound~$k$ on the assumption.

While in the case of $\ltl$ the assume-guarantee problem $\tuple{\varphi,\psi}$ can be reduced to the $\ltl$ realizability problem for the implication $\varphi \rightarrow \psi$, this is not possible in $\prompt$ due to the  quantifier alternation on the bounds.
As a matter of fact, we do not know yet whether the $\prompt$ assume-guarantee realizability problem in the single-process case is decidable.
We show that even if the problem would turn out to be decidable, an implementation that realizes the specification in general may need infinite memory.

\begin{lemma}
   There exists a $\prompt$ assume-guarantee specification that can be realized with an infinite-state strategy, but not with a finite-state one.
\end{lemma}
\begin{proof}
  Consider the assume-guarantee specification $\tuple{\varphi,\psi}$ with $\varphi = \GFp o \lor \FG \neg o$ and $\psi = \false$ and a single process architecture with $I = \emptyset$ and $O = \set{o}$.
  As the guarantee $\psi$ is false, the implementation has to falsify the assumption~$\varphi$ for every bound $k$ on the prompt-eventually operator to realize $\tuple{\varphi,\psi}$.
  To falsify $\varphi$ with respect to a fixed $k$, the implementation has to produce a sequence $w \in (2^\set{o})^\omega$ where $o$ is true infinitely often and where $\emptyset^k$ is an infix  of $w$.
  Thus, the size of the implementation depends on $k$ and an implementation that falsifies $\varphi$ for every $k$ must have infinite memory.
\end{proof}
Moreover, already the $\ltl$ realizability problem is undecidable in the asynchronous case. Thus, the $\prompt$ assume-guarantee realizability problem for asynchronous architectures may be at best solvable by a semi-decision procedure. 
We present such a semi-decision procedure for the asynchronous $\prompt$ assume-guarantee realizability problem based on bounded synthesis~\cite{journals/sttt/FinkbeinerS13}.
In bounded synthesis, a transition system of a fixed size is ``guessed'' and model checked by a constraint solver.
Model checking for $\prompt$ can be solved by checking pumpable non-emptiness of colored B\"uchi graphs~\cite{journals/fmsd/KupfermanPV09}.
However, the pumpability condition cannot directly be expressed in the bounded synthesis constraint system.
Hence, in Section~\ref{sec:colored-buchi-graphs}, we give an alternative solution to the non-emptiness of colored B\"uchi graphs by a reduction to B\"uchi graphs that have access to the state space of the transition system.
We show how to extend bounded synthesis to such B\"uchi graphs in Section~\ref{sec:bounded-synthesis}, and present a semi-decision procedure for $\prompt$ assume-guarantee synthesis based on this extension in Section~\ref{sec:semi-algorith-ag}.

In the following we use transition systems as representations for finite-state strategies, since the algorithm developed in this section needs access to the syntactic representation of strategies.

\subsection{Nonemptiness of Colored B\"uchi Graphs} \label{sec:colored-buchi-graphs}

In the case of $\ltl$ specifications, the nonemptiness problem for B\"uchi graphs gives a classical solution to the model checking problem for a given system $\tsys$.
Let $\varphi$ be the $\ltl$ formula that $\tsys$ should satisfy.
In a preprocessing step, the negation of $\varphi$ is translated to a nondeterministic B\"uchi word automaton $\nbw_{\neg \varphi}$~\cite{BaierKatoen08}.
Then, $\varphi$ is violated by $\tsys$ if, and only if, the B\"uchi graph $G$ representing the product of $\tsys$ and $\nbw_{\neg\varphi}$ is nonempty.
An accepting path $\pi$ in $G$ witnesses a computation of $\mathcal{S}$ that violates~$\varphi$.
\emph{Colored B\"uchi graphs} are an extension to such graphs in the context of model checking $\prompt$~\cite{journals/fmsd/KupfermanPV09}.

A colored B\"uchi graph of degree two is a tuple $G = \tuple{\set{r,r'},V,E,v_0,L,\Bwin}$, where
$r$ and $r'$ are propositions,
$V$ is a set of vertices,
$E \subseteq V \times V$ is a set of edges,
$v_0 \in V$ is the designated initial vertex,
$L\colon\fun{V}{\pow{\set{r,r'}}}$ describes the color of a vertex, and
$\Bwin = \set{\bwin_1,\bwin_2}$ is a generalized B\"uchi condition of index two, i.e., $\bwin_1, \bwin_2 \subseteq V$.
A B\"uchi graph is a special case where we omit the labeling function and are interested in finding an accepting path.
A path $\pi = v_0 v_1 v_2 \cdots \in V^\omega$ is pumpable if we can pump all its $r'$-blocks without pumping its $r$-blocks.
Formally, a path is pumpable if for all adjacent $r'$-change points $i$ and $i'$, there are positions $j$, $j'$, and $j''$ such that $i \leq j < j ' < j'' < i'$, $v_j = v_{j''}$ and $r \in L(v_j)$ if, and only if,  $r \notin L(v_{j'})$.
A path $\pi$ is accepting, if it visits both $\bwin_1$ and $\bwin_2$ infinitely often.
The \emph{pumpable nonemptiness} problem for $G$ is to decide whether $G$ has a pumpable accepting path.
It is $\nlogspace$-complete and solvable in linear time~\cite{journals/fmsd/KupfermanPV09}.

We give an alternative solution to this problem based on a reduction to the nonemptiness problem of B\"uchi graphs.
To this end, we construct a non-deterministic safety automaton $\nbw_\text{pump}$ that characterizes the pumpability condition.
A non-deterministic safety automaton $\nbw$ is a tuple $\tuple{\Sigma,S,s_0,\delta}$, where $\Sigma$ is a finite alphabet, $S$ is a finite set of states, $s_0 \in S$ is the designated initial state, and $\delta \colon \fun{S \times \Sigma}{\pow{S}}$ is the transition function.
An infinite word is accepted by a safety automaton $\nbw$ if, and only if, there exists an infinite run on this word.
\begin{lemma} \label{thm:buchi_pumpable}
  Let $\cbgraph = \tuple{\set{r,r'},V,E,v_0,L,\Bwin}$ be a colored B\"uchi graph of degree two.
  There exists a B\"uchi graph $\bgraph'$, with $\bigo(\card{\bgraph'}) = \bigo(\card{\cbgraph}^2)$, such that $\cbgraph$ has a pumpable accepting path if, and only if, $\bgraph'$ has an accepting path.
\end{lemma}
\begin{proof}
  We define a non-deterministic safety automaton $\nbw_\text{pump} = \tuple{V \times 2^\set{r,r'},S,s_0,\delta}$ over the alphabet $V \times \pow{\set{r,r'}}$ that checks the pumpability condition.
  The product of $\cbgraph$ and $\nbw_\text{pump}$ (defined later) represents the B\"uchi graph $G'$ where every accepting path is pumpable.
  
  The language $\mathcal{L} \subseteq (V \times \pow{\set{r,r'}})^\omega$ of pumpable paths (with respect to a fixed set of vertices $V$) is an $\omega$-regular language that can be recognized by a small non-deterministic safety automaton.
  This automaton~$\nbw_\text{pump}$ operates in 3 phases between every pair of adjacent $r'$-change points:
  first, it non-deterministically remembers a vertex $v$ and the corresponding truth value of $r$.
  Then, it checks that this value changes and thereafter it remains to show that the vertex $v$ repeats before the next $r'$-change point.
  Thus, the state space~$S$ of $\nbw_\text{pump}$ is
  \begin{align*}
    \set{s_0} \cup \Set{s_{v,x} \mid v \in V, x \in 2^\set{r,r'}}
    \cup \Set{s'_{v,y} \mid v \in V, y \in 2^\set{r,r'}}
    \cup \Set{s''_z \mid z \in 2^\set{r'}}
  \end{align*}
  and the initial state is $s_0$.
  The state space corresponds to the 3 phases: in the states $s_{v,x}$ a vertex $v$ and a truth value of $r$ are remembered, before state $s'_{v,y}$ the value of $r$ changes, and $s''_z$ is the state after the vertex repetition.
  The transition function $\delta \colon (S \times (V \times 2^\set{r,r'})) \rightarrow 2^S$ is defined in the following.
  We use the notation $A =_C B$ to denote $(A \cap C) = (B \cap C)$.
  \begin{itemize}
    \item $\delta(s_0,(v,x)) =  \set{ s_{v,x} }$
    \item $\delta(s_{v,x},(v',x')) \ni \begin{cases}
      s_{v,x} & \text{if } x =_\set{r'} x' \\
      s_{v',x'} & \text{if } x =_\set{r'} x' \\
      s'_{v,x'} & \text{if } x =_\set{r'} x' \text{ and } x \neq_\set{r} x'
    \end{cases}$
    \item $\delta(s'_{v,y},(v',x)) \ni \begin{cases}
      s'_{v,y} & \text{if } x =_\set{r'} y \text{ and } v' \neq v \\
      s''_{y \cap \set{r'}} & \text{if } x =_\set{r'} y \text{ and } v' = v
    \end{cases}$
    \item $\delta(s''_z, (v,x)) \ni \begin{cases}
      s''_z & \text{if } x =_\set{r'} z \\
      s_{v,x} & \text{if } x \neq_\set{r'} z
    \end{cases}$
  \end{itemize}
  The size of $\nbw_\text{pump}$ is in $O(\card{V})$.
  Figure~\ref{fig:safety_automaton_pumpable} gives a visualization of this automaton.
  
  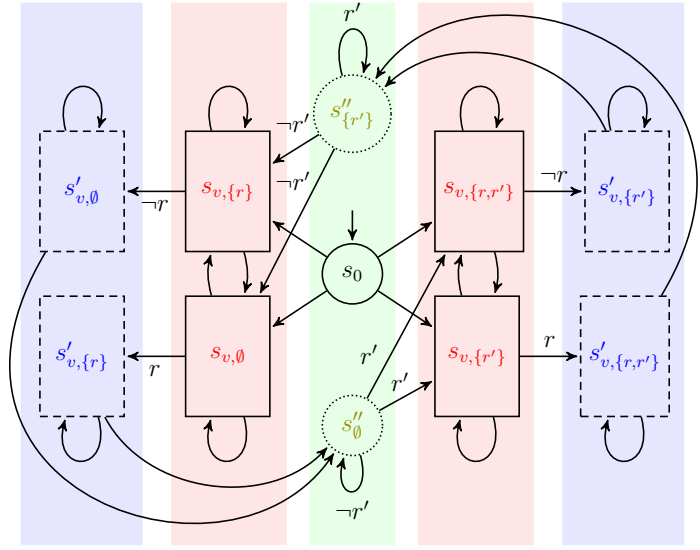
\begin{figure}[t]
    \centering
    \begin{tikzpicture}[->,>=stealth',shorten >=1pt,auto,node distance=1cm,semithick,scale=0.8,transform shape]

  \draw[fill,color=red!10] (-1.1,4.5) rectangle (-3,-4.5);
  \draw[fill,color=red!10] (1.1,4.5) rectangle (3,-4.5);
  
  \draw[fill,color=blue!10] (-3.5,4.5) rectangle (-5.5,-4.5);
  \draw[fill,color=blue!10] (3.5,4.5) rectangle (5.5,-4.5);
  
  \draw[fill,color=green!10] (-0.7,4.5) rectangle (0.7,-4.5);

  \tikzstyle{block}=[state,rectangle,minimum width=1.4cm,minimum height=2cm]
  \tikzstyle{nonblock}=[state,minimum size=1cm]

  \node[nonblock,initial above,initial text=] (init) {$s_0$};
  
  \node[block,above right=0 and 1 of init] (s_qp) {\color{red}$s_{v,\set{r,r'}}$};
  \node[block,below right=0 and 1 of init] (s_q) {\color{red}$s_{v,\set{r'}}$};
  \node[block,above left=0 and 1 of init] (s_p) {\color{red}$s_{v,\set{r}}$};
  \node[block,below left=0 and 1 of init] (s_) {\color{red}$s_{v,\emptyset}$};
  
  \node[block,right=of s_qp,densely dashed] (sp_q) {\color{blue}$s'_{v,\set{r'}}$};
  \node[block,right=of s_q,densely dashed] (sp_qp) {\color{blue}$s'_{v,\set{r,r'}}$};
  \node[block,left=of s_p,densely dashed] (sp_) {\color{blue}$s'_{v,\emptyset}$};
  \node[block,left=of s_,densely dashed] (sp_p) {\color{blue}$s'_{v,\set{r}}$};
  
  \node[nonblock,above=1.5 of init,densely dotted] (spp_q) {\color{olive}$s''_\set{r'}$};
  \node[nonblock,below=1.5 of init,densely dotted] (spp_) {\color{olive}$s''_\emptyset$};
  
  \draw (init) edge (s_qp)
        (init) edge (s_q)
        (init) edge (s_p)
        (init) edge (s_)
        
        (s_qp) edge[bend left=15] (s_q)
        (s_q) edge[bend left=15] (s_qp)
        (s_p) edge[bend left=15] (s_)
        (s_) edge[bend left=15] (s_p)
        
        (s_qp) edge node {$\neg r$} (sp_q)
        (s_q) edge node {$r$} (sp_qp)
        (s_p) edge node {$\neg r$} (sp_)
        (s_) edge node {$r$} (sp_p)
        
        (sp_q) edge[bend right=55] (spp_q)
        (sp_p) edge[bend right=55] (spp_)
        
        (spp_q) edge[loop above] node {$r'$} ()
        (spp_) edge[loop below] node {$\neg r'$} ()
        
        (spp_q) edge node[pos=0,swap,yshift=-3pt] {$\neg r'$} (s_p)
        (spp_q) edge node[pos=0.2,swap,yshift=-10pt] {$\neg r'$} (s_)
        (spp_) edge node[near start,yshift=-5pt] {$r'$} (s_qp)
        (spp_) edge node[pos=0.7,yshift=-5pt] {$r'$} (s_q)
        
        (s_qp) edge[loop above,min distance=10mm] (s_qp)
        (s_p) edge[loop above,min distance=10mm] (s_p)
        (s_q) edge[loop below,min distance=10mm] (s_q)
        (s_) edge[loop below,min distance=10mm] (s_)
        
        (sp_qp) edge[loop below,min distance=10mm] (sp_qp)
        (sp_p) edge[loop below,min distance=10mm] (sp_p)
        (sp_q) edge[loop above,min distance=10mm] (sp_q)
        (sp_) edge[loop above,min distance=10mm] (sp_)        
        ;
  
  \draw (sp_qp) .. controls +(3,5) and +(2,3) .. (spp_q);
  \draw (sp_) .. controls +(-3,-5) and +(-2,-3) .. (spp_);
  
\end{tikzpicture}
    \caption{Schematic visualization of the automaton $\nbw_\text{pump}$ from the proof of Lemma~\ref{thm:buchi_pumpable}. The 3 phases are clearly visible: In the red states {\color{red}$s_{v,x}$} (solid rectangles) the values $(v,x)$ are non-deterministically stored and those states can only be left if there is a change in the value of $r$. The subsequent blue states {\color{blue}$s'_{v,y}$} (dashed rectangles) can only be left in case of a vertex repetition leading to the green state {\color{olive}$s''_z$} (dotted circles) that waits for the next $r'$-change point.}
    \label{fig:safety_automaton_pumpable}
  \end{figure}
  
  We define the product $G'$ of the colored B\"uchi graph $\cbgraph = \tuple{\set{r,r'},V,E,v_0,L,\Bwin}$ and the automaton~$\nbw_\text{pump}$ as the B\"uchi graph $(V \times S,E',(v_0,s_0),\Bwin')$, where 
  \begin{equation*}
    ((v,s),(v',s')) \in E' \quad\Leftrightarrow\quad (v,v') \in E \land s' \in \delta(s,(v, L(v)))  
  \end{equation*}
  and where $\Bwin' = (\bwin'_1,\bwin'_2)$ is given by  $\bwin'_i = \set{(v,s) \mid v \in \bwin_i \text{ and } s \in S}$ for  $i \in \set{1,2}$.
  The size of $G'$ is in $\bigo(\card{G}^2)$. It remains to show that $G$ has a pumpable accepting path if, and only if, $G'$ has an accepting path.

  Consider a pumpable accepting path $\pi$ in $G$.
  We show that there is a corresponding accepting path $\pi'$ in $G'$.
  Let $i$ and $i'$ be adjacent $r'$-change points.
  Then, there are positions $j$, $j'$, and $j''$ such that $i \leq j < j ' < j'' < i'$, $v_j = v_{j''}$ and $r \in L(v_j)$ if, and only if,  $r \notin L(v_{j'})$.
  By construction, at position $i$, automaton $\nbw_\text{pump}$ is at some state from the set $\set{s_0,s''_\emptyset,s''_\set{r'}}$.
  We follow the automaton and remember vertex $v$ and the truth value of $r$ at position $j \geq i$ (some state $s_{v,x}$).
  Next, we take the transition to $s'_{v,y}$ where the truth value of $r$ changes (at position $j'$).
  Lastly, we check that there is a vertex repetition (at position $j''$) and go to state $s''_z$.
  At the next $r'$-change point $i'$, the argument repeats.
  This path is accepting, as the original one is accepting.
  
  Now, consider an accepting path $\pi$ in $G'$.
  We show that there is a pumpable accepting path in $G$.
  Let $\pi'$ be the projection of every position of $\pi$ to the first component.
  By construction, $\pi'$ is an accepting path in $G$.
  Let $\pi_i \pi_{i+1} \cdots \pi_{i'}$ be an $r'$-block of $\pi$.
  As $\pi$ has a run of the automaton $\nbw_\text{pump}$, we know that there exists a state repetition between $i$ and $i'$ where the truth value of $r$ changes in between.
  Hence, the path $\pi'$ is pumpable. 
\end{proof}
\begin{remark}
  Note that in the context of the previous proof, it would be enough to remember a vertex $v$ without the valuation of $\set{r,r'}$, as the vertex determines the valuation by the labeling function $L\colon\fun{v}{2^\set{r,r'}}$ of $\cbgraph$.
  However, we will later use $\nbw_\text{pump}$ in a more general setting (cf.~Section~\ref{sec:semi-algorith-ag}).
\end{remark}

\subsection{Bounded Synthesis} \label{sec:bounded-synthesis}
%
Bounded synthesis~\cite{journals/sttt/FinkbeinerS13} is a semi-decision procedure for the distributed synthesis problem. In its original form, it takes as input a specification expressed by a universal co-B\"uchi automaton $\ubuchi$, a (possibly asynchronous) architecture $\arch$, and a size bound $\bound$ (or a family of bounds on the individual processes), and decides whether a correct implementation of the given size exists.
Bounded synthesis expresses the acceptance of a transition system $\tsys$ on $\ubuchi$, i.e., acceptance of all traces generated by $\tsys$, as a constraint system in a decidable first-order theory.
In this section, we show a modification of bounded synthesis that gives the specification automaton access to the states of the system to be synthesized.
This extension is needed for automata that can express the pumpability condition, in particular the one we constructed in the proof of Lemma~\ref{thm:buchi_pumpable}. 
We will show in Section~\ref{sec:semi-algorith-ag} how to obtain such an automaton from a $\prompt$ assume-guarantee specification~$\tuple{\varphi,\psi}$, resulting in a semi-decision procedure for asynchronous distributed synthesis from this class of specifications.

For distributed architectures, bounded synthesis separately considers the 
problems of finding a global transition system that is accepted by $\ubuchi$ 
and of dividing the transition system into local components according to the 
given architecture. To this end, two sets of constraints are generated: (i) 
an encoding of the acceptance by $\ubuchi$ of a global transition system 
$\tsys$ of size $\bound$, and (ii) an encoding of the architectural 
constraints that divides this global system into local components. If the 
conjunction of both sets of constraints is satisfiable, then a satisfying assignment of the 
constraints represents a distributed system that satisfies $\varphi$ in $\arch
$. Since the architectural constraints we consider are the same as in 
standard bounded synthesis, we only have to modify the constraints encoding 
the existence of a global transition system that satisfies the given 
specification.

\paragraph{Extended Automata}
We define a \emph{universal co-B\"uchi tree automaton} as a tuple~$\ubuchi = \tuple{\Sigma,\Upsilon,Q,q_0,\delta,\cobwin}$, where 
$\Sigma$ is an input alphabet, 
$\Upsilon$ is a set of directions, 
$Q$ is a set of states,
$\delta\colon Q \times \Sigma \rightarrow 2^{Q \times \Upsilon}$ is the transition function, and
$\cobwin \subseteq Q$ is the set of rejecting states.

As mentioned above, we want to check acceptance of a \emph{global} transition system by $\ubuchi$. Therefore, we consider the sets of inputs $I = O^*_{\penv}$ and outputs $O = \bigcup_{p \in P^-} O^*_p $ of the composition of the system processes, and are interested in the acceptance of a $\pow{O}$-labeled $\pow{I}$-transition system $\tsys = \tuple{S,s_0,\Delta,l}$. In addition, we want to recognize the pumpability condition. Therefore, we consider a \emph{state-aware} universal co-B\"uchi tree automaton with $\Sigma = \pow{O} \times S$ and $\Upsilon = \pow{I}$, i.e., in addition to output valuations, the automaton has access to the current state of $\tsys$.

Acceptance of $\tsys$ by the automaton is defined in terms of run graphs: the \emph{run graph} of an automaton $\ubuchi_S = \tuple{\pow{O} \times S,\pow{I},Q,q_0,\delta,\cobwin}$ on $\tsys$ is the minimal directed graph $\rungraph = (G,E)$ that satisfies the constraints:
\begin{itemize}
\item $G \subseteq Q \times S$,
\item $(q_0,s_0) \in G$, and
\item for every $(q,s) \in G$, it holds 
$$\left\{ (q',\upsilon) \in Q \times \pow{I} \mid \left( (q,s), (q',\Delta(s,\upsilon)) \right) \in E \right\} \supseteq \delta(q,(l(s),s)).$$
\end{itemize}

The \emph{co-B\"uchi condition} requires that, for an infinite path $g_0 g_1 g_2 \cdots \in G^\omega$ of the run graph, $g_i \in \cobwin \times S$ holds for only finitely many $i \in \nats$.
A run graph is \emph{accepting} if every infinite path $g_0 g_1 g_2 \cdots \in G^\omega$ of the run graph satisfies the co-B\"uchi condition.
A transition system $\tsys$ is \emph{accepted by $\ubuchi_S$} if the (unique) run graph of $\ubuchi_S$ on $\tsys$ is accepting.

\paragraph{Annotated transition systems}
We introduce an annotation function for transition systems that witnesses 
acceptance by a (possibly state-aware) universal co-B\"uchi tree automaton. 
The annotation assigns to each pair $(q,s) \in Q \times S$ a natural number 
or a special symbol $\bot$. Natural numbers indicate the maximal number of 
occurrences of rejecting states on any path to $(q,s)$ in the run graph;
$\bot$ indicates that the pair $(q,s)$ is not reachable.
Thus, if for a given transition system there exists an annotation that assigns natural 
numbers to all vertices of the run graph, then the number of 
visits to rejecting states must be bounded in any run. 
Such annotations are called \emph{valid}, and 
transition systems with valid annotations are exactly those that are accepted 
by the automaton.

An \emph{annotation} of a $\pow{O}$-labeled $\pow{I}$-transition system $\tsys = \tuple{S,s_0,\Delta,l}$ on a state-aware universal co-B\"uchi tree automaton $\ubuchi_S = \tuple{\pow{O} \times S,\pow{I},Q,q_0,\delta,\cobwin}$ is a function $\lambda\colon Q \times S \rightarrow \{\bot\} \cup \nats$. An annotation is \emph{valid} if it satisfies the following conditions:
\begin{itemize}
\item $\lambda(q_0,s_0) \neq \bot$
\item for any $(q,s) \in Q \times S$:
\begin{itemize}
\item[] if $\lambda(q,s) = n \neq \bot$ and $(q',\upsilon) \in \delta(q,(l(s),s))$ 
\item[] then $\lambda(q',\Delta(s,\upsilon)) \vartriangleright \lambda(q,s)$, 
\item[] where $\vartriangleright$ is interpreted as $>$ if $q' \in \cobwin$, and $\geq$ otherwise.
\end{itemize}
\end{itemize}\noindent
An annotation is \emph{$c$-bounded} if its codomain is contained in $\{\bot, 0, \ldots, c\}$.

\begin{theorem}[see \cite{journals/sttt/FinkbeinerS13}] \label{thm:bounded-annotation}
A finite-state $O$-labeled $I$-transition system $\tsys = \tuple{S,s_0,\Delta,l}$ is accepted by a state-aware universal co-B\"uchi tree automaton $\ubuchi_S = \tuple{\pow{O} \times S,\pow{I},Q,q_0,\delta,\cobwin}$ if, and only if, it has a valid $(\card{S} \cdot \card{\cobwin})$-bounded annotation.
\end{theorem}

\begin{proof}
The original proof by Finkbeiner and Schewe~\cite{journals/sttt/FinkbeinerS13} works without modifications for our extension to state-aware universal co-B\"uchi tree automata.
\end{proof}

For a given state-aware universal co-B\"uchi tree automaton $\ubuchi_S = 
\tuple{\pow{O} \times S,\pow{I},Q,q_0,\delta,\cobwin}$, 
Theorem~\ref{thm:bounded-annotation} allows us to decide the existence of an
 $O$-labeled $I$-transition system with state space $S$ that is accepted by 
$\ubuchi_S$. 

\paragraph{Encoding of global acceptance}
The existence of a (global)
transition system with a valid annotation can be encoded into 
a set of decidable constraints in first-order logic modulo a theory with uninterpreted functions and a partial order. Essentially, we can directly 
encode the conditions for a valid annotation into constraints, with uninterpreted transition function and labeling for the desired transition system. Such constraints can then be solved by off-the-shelf satisfiability modulo theories (SMT) tools.
  
Like the proof of 
Theorem~\ref{thm:bounded-annotation}, the original encoding can easily be extended to support
our notion of state-aware universal B\"uchi tree automata. It is constructed in the following way:
\begin{enumerate}
\item Assume that $\ubuchi_S$ is defined in a suitable way, i.e., the sets $Q$ and $\cobwin$, state $q_0$, and transition relation $\delta\colon Q \times (2^O \times S) \rightarrow 2^{(Q \times 2^I)}$ are defined.
\item Declare uninterpreted sets and functions for the transition system $\tsys$ and the annotation:
\begin{itemize}
\item Define the set of states $S$ as $\{1,\ldots,b\}$ for the given bound $b \in \nats$.
\item Declare the transition function of $\tsys$ as $\Delta: S \times \pow{I} \rightarrow S$ and the labeling function as $l: S \rightarrow \pow{O}$.
\item Declare two functions that are used to model the annotation function: $\lambda^\bool: Q \times S \rightarrow \bool = \set{\true,\false}$ and $\lambda^\#: Q \times S \rightarrow  \nats$.
\end{itemize}
\item Assert the following constraints:
$$\begin{array}{ll}
				&	s_0 \in S\\[5pt]
				& \lambda^\bool(q_0,s_0)\\[5pt]
\forall q,q' \in Q, s \in S, \upsilon \in \pow{I}: & \lambda^\bool(q,s) \land (q',\upsilon) \in \delta(q,(l(s),s))\\
				& \rightarrow \lambda^\bool(q',\Delta(s,\upsilon)) \land \lambda^\#(q',\Delta(s,\upsilon)) \geq \lambda^\#(q,s)\\[5pt]
\forall q,q' \in Q, s \in S, \upsilon \in \pow{I}: & \lambda^\bool(q,s) \land (q',\upsilon) \in \delta(q,(l(s),s)) \land q' \in \cobwin\\
				& \rightarrow \lambda^\#(q',\Delta(s,\upsilon)) > \lambda^\#(q,s)\\[5pt]
\end{array}$$
\end{enumerate}

The encoding ensures that $\lambda^\bool(q,s)$ is true whenever $(q,s) \in Q \times S$ is reachable in the run graph of $\ubuchi_S$ on $\tsys$, and that $\lambda^\#$ respects the conditions for a valid annotation for all reachable vertices $(q,s)$. Since there are no conditions for the annotation on vertices that are not reachable, a solution for $\lambda^\#$ will represent a valid annotation of $\tsys$ on $\ubuchi_S$.

Note that our encoding is a strict generalization of the encoding of Finkbeiner and Schewe~\cite{journals/sttt/FinkbeinerS13}. In particular, our encoding can also be used for specifications in $\ltl$ that are translated into a universal co-B\"uchi tree automaton (see Kupferman and Vardi~\cite{KupfermanV05}), which can be seen as a state-aware automaton that ignores the state of the transition system.

\paragraph{Encoding of architectural constraints}
As mentioned above, the encoding of architectural constraints can be adopted from the original approach without changes. For a given asynchronous architecture $\arch^* = \tuple{P,\penv,\{I^*_p\}_{p \in P}, \{O^*_p\}_{p \in P}}$, the additional constraints (1) assert that the state of a process $p \in P^-$ does not change if it is not scheduled and (2) that the transitions of a process only depend on its current state and the visible inputs. In addition, it can contain additional bounds on the state space of every single component.

The conjunction of both sets of constraints then asks for the existence of a distributed implementation $\tsys = \Distprod_{p \in P^-} \tsys_p$ of size $\bound$ that is accepted by $\ubuchi$, possibly with additional bounds $b_{p}$ for every $p \in \pminus$ on the size of the components. 

\begin{theorem}[see \cite{journals/sttt/FinkbeinerS13}]
\label{thm:asynchronous-bounded-synthesis}
Given a state-aware universal co-B\"uchi tree automaton $\ubuchi_S$~\footnote{As the symbol $S$ in $\ubuchi_S$ refers to the state-space of the distributed product, $\card{S}$ has to be equal to the product of bounds $b_p$ for $p \in \pminus$.}, an asynchronous architecture $\arch^*$, and a family of bounds $b_{p}$ for every $p \in \pminus$, there is a constraint system (in a decidable first-order theory) that is satisfiable if, and only if, there exist implementations $\tsys_{p}$ of size~$b_p$ for every $p \in \pminus$ such that $\tsys = \Distprod_{p \in P^-} \tsys_p$ is accepted by $\ubuchi_S$ and satisfies the architectural constraints of $\arch^*$.
\end{theorem}

\begin{proof}
Follows immediately from Theorem~\ref{thm:bounded-annotation} and the correctness of the architectural constraints from Finkbeiner and Schewe~\cite{journals/sttt/FinkbeinerS13}.
\end{proof}

\subsection{A Semi-Decision Procedure for Assume-Guarantee Realizability} \label{sec:semi-algorith-ag}
Since the assume-guarantee realizability problem for asynchronous architectures is undecidable and infinite-state strategies are required in general, we give a semi-decision procedure for the problem. Our solution is based on the techniques developed in the last subsections.

As the bounded synthesis approach described in the last subsection already accounts for ``guessing'' transition systems $\tsys_p$ for each system process $p$ according to the architectural constraints given by $\arch^*$, we reduce the problem of model checking individual implementations $\tsys_p$ to model checking the product system $\tsys = \Distprod_{p \in \pminus} \tsys_p$.
A transition system $\tsys$ satisfies an assume-guarantee specification $\tuple{\varphi,\psi}$ if the strategy $f$ generated by $\tsys$ satisfies $\tuple{\varphi,\psi}$, i.e., if for every bound~$k$ there is a bound $l$ such that for every $w \in f$, we have that $(w, k) \models \varphi$ implies $(w, l) \models \psi$.

Given an assume-guarantee specification $\tuple{\varphi,\psi}$,
we first solve the problem of model checking assume-guarantee specifications by building a state-aware universal co-B\"uchi tree automaton $\uct$ that accepts a transition system $\tsys$ if, and only if, $\tsys$ satisfies $\tuple{\varphi,\psi}$.
Given $\uct$ and a bound $\bound$ on the size of the implementation, we can then use the encoding from Section~\ref{sec:bounded-synthesis} to decide realizability modulo this bound, and obtain a semi-decision procedure by solving the problem for increasing bounds.

\paragraph{Encoding $\tuple{\varphi,\psi}$ into B\"uchi automata}
Let $\arch^* = \tuple{P,\penv,\{I^*_p\}_{p \in P}, \{O^*_p\}_{p \in P}}$ be  an asynchronous architecture and let $I = O^*_{\penv}$ and $O = \bigcup_{p \in P^-} O^*_p $ be the set of inputs and outputs, respectively, of the composition of the system processes. 
First, we construct the non-deterministic B\"uchi  automaton $\nbw_{\overline{c}_{r'}(\psi) \land c_r(\varphi)} = \tuple{2^{I \cup O  \cup \set{r,r'}}, Q, q_0, \delta, \bwin}$, where $\overline{c}_{r'}(\psi) = \alt_{r'} \land \neg\rel_{r'}(\psi)$, whose language contains exactly those paths that satisfy ${\overline{c}_{r'}(\psi) \land c_r(\varphi)}$~\cite{BaierKatoen08}.
Then, we use the following lemma to characterize whether a transition system $\tsys$ satisfies an assume-guarantee specification $\tuple{\varphi,\psi}$ by reducing it to finding pumpable error paths in the two-color B\"uchi graph $G = \tuple{\set{r,r'},V,E,v_0,L,\Bwin}$, as introduced in Section~\ref{sec:colored-buchi-graphs}, which is the product of $\tsys = \tuple{S,s_0,\Delta,l}$ and $\nbw_{\overline{c}_{r'}(\psi) \land c_r(\varphi)}$.
Formally, the elements of $G$ are defined as
$V = S \times 2^\set{r,r'} \times Q$,
$E$ as $((s,R,q),(s',R',q')) \in E$ if, and only if, there is an input valuation $\vec{i} \in 2^I$ such that $s' = \Delta(s,\vec{i})$ and $(q',\vec{i}) \in \delta(q,l(s))$,
$v_0 = (s_0, \emptyset, q_0)$,
$L$ as $L((s,R,q,q^*)) = R$, and
$\Bwin = \set{\bwin}$.
\begin{lemma}
\label{lemma_pumpabilitycharac}
	Let $\tuple{\varphi,\psi}$ be a $\prompt$ assume-guarantee specification, $\arch^*$ an asynchronous architecture and $\tsys_p$ a finite-state implementation for every system process $p \in \pminus$.
	The distributed product $\tsys = \Distprod_{p \in \pminus} \tsys_p$ does not satisfy $\tuple{\varphi,\psi}$ if, and only if, the product of $\tsys$ and $\nbw_{\overline{c}_{r'}(\psi) \land c_r(\varphi)}$ is pumpable non-empty.
\end{lemma}
\begin{proof}
Similar to the proof of	Lemma~6.1 and Theorem~6.2 in \cite{journals/fmsd/KupfermanPV09}. The missing proof of Lemma~6.1 is given in \cite{journals/iandc/FaymonvilleZ17} (Lemma~8). See also the discussion below its proof.
\end{proof}

%
%
To check the existence of pumpable error paths, we use the non-deterministic automaton~$\nbw_\text{pump} = \tuple{V \times 2^\set{r,r'},S,s_0,\delta',S}$ from the proof of Lemma~\ref{thm:buchi_pumpable}. Here, we let $V = X \times Q$, where $X$ is a set with $\bound$ elements, representing the state space of the desired solution $\tsys$, and $Q$ is the state space of the automaton~$\nbw_{\overline{c}_{r'}(\psi) \land c_r(\varphi)}$ defined above, that is, we use as $V$ the state space $X \times Q$ of the colored B\"uchi graph that is used to model check an implementation $\tsys$ against a specification $\tuple{\psi,\varphi}$. 

The product of $\nbw_{\overline{c}_{r'}(\psi) \land c_r(\varphi)}$ and $\nbw_\text{pump}$ is an automaton $\nbw$ that operates on the inputs $I$, outputs $O$, propositions $\set{r,r'}$, and the state space $X$ of the implementation, and accepts all those paths that are pumpable and violate the assume-guarantee specification (cf. Lemma~\ref{thm:buchi_pumpable}). 

Formally, $\nbw$ is defined as:
$$\tuple{2^{I  \cup O \cup \set{r,r'}} \times X, Q \times S, (q_0,s_0),\delta^*,\bwin^*},$$
where $\delta^* \colon Q \times S \times 2^{I \cup O \cup \set{r,r'}} \times \set{x} \rightarrow 2^{Q \times S}$ is defined as 
$$\delta^*((q,s),(\sigma,x)) = \left\{ (q',s') ~\mid~ q' \in \delta(q, \sigma) ~\wedge~ s' \in \delta'(s, \set{q,x} \cup (\sigma\cap\set{r,r'})) \right\},$$
and $\bwin^*$ is the B\"uchi condition $\set{(q,s) \mid q \in \bwin, s \in S}$.

We complement $\nbw$, resulting in a universal co-B\"uchi automaton $\ucw$ that accepts a given sequence $w \in (2^{I \cup  \set{r,r'}})^\omega$ of inputs and the behavior of an implementation $\tsys$ on $w$ if, and only if, the execution of $\tsys$ on $w$ satisfies $\tuple{\psi,\varphi}$. Finally, we construct a (state-aware) universal co-B\"uchi tree automaton $\uct = (2^O \times X, 2^{I \cup \set{r,r'}},Q,q_0,\delta,\cobwin)$ by spanning a copy of $\ucw$ for every direction in $2^{I \cup \set{r,r'}}$.
Then, an implementation $\tsys$ with set~$S$ of states is accepted by $\uct$ if, and only if, $\tsys$ satisfies $\tuple{\varphi,\psi}$ (for all possible input sequences). Thus,  
$\uct$ solves the problem of model checking assume-guarantee specifications.

\paragraph{Encoding the automaton into constraints}
Now, we can use the modified bounded synthesis algorithm from Section~\ref{sec:bounded-synthesis} to encode $\uct$ into a set of constraints that is satisfiable if, and only if, there exists an implementation~$\tsys$ that satisfies $\tuple{\varphi,\psi}$.
We obtain the following corollaries stating the correctness of the constraint system for single-process implementations (Corollary~\ref{thm:bs_correct_global}) and distributed implementations (Corollary~\ref{thm:bs_correct_local}), respectively.

\begin{corollary} \label{thm:bs_correct_global}
Given a $\prompt$ assume-guarantee specification $\tuple{\varphi,\psi}$ and a bound $b$, there is a constraint system (in a decidable first-order theory) that is satisfiable if, and only if, there exists an implementation $\tsys$ of size~$b$ such that $\tsys$ satisfies $\tuple{\varphi,\psi}$.
\end{corollary}

\begin{corollary} \label{thm:bs_correct_local}
  Given a $\prompt$ assume-guarantee specification $\tuple{\varphi,\psi}$, an asynchronous architecture $\arch^*$, and a family of bounds $b_{p}$ for each $p \in \pminus$, there is a constraint system (in a decidable first-order theory) that is satisfiable if, and only if, there exist implementations $\tsys_{p}$ of size~$b_p$ for each $p \in \pminus$ such that $\Distprod_{\tsys \in P^-} \tsys_p$ satisfies $\tuple{\varphi,\psi}$ in $\arch^*$.
\end{corollary}

By exhaustively traversing the space of bounds~$(b_{p})_{p \in \pminus}$ and by solving the resulting constraint system, we obtain a semi-decision procedure for the asynchronous $\prompt$ assume-guarantee realizability problem. 
Furthermore, this also solves the synthesis problem, as a satisfying assignment of the constraint system directly represents a valid implementation, where the transition relation is given by (the assignment of) the function $\Delta$, and the state labeling by (the assignment of) the function $l$.

\begin{corollary}
  Let $\arch^*$ be an asynchronous architecture.
  The asynchronous $\prompt$ assume-guarantee realizability problem for $\arch^*$ is semi-decidable.
\end{corollary}

\section{Beyond $\prompt$} \label{sec:beyond_prompt_ltl}

In this section, we consider distributed synthesis for logics stronger than $\prompt$. As already pointed out in the introduction, $\prompt$ is predated by parametric linear temporal logic ($\pltl$), which was introduced by Alur et al.~\cite{journals/tocl/AlurETP01}. This logic is obtained by adding parameterized eventually operators of the form~$\F_{\le x}\phi$ and parameterized always operators of the form~$\G_{\le y}$ to $\ltl$. Here, $x$ and $y$ are variables which are instantiated by a variable valuation~$\alpha$ mapping variables to natural numbers that serve as bounds: $\F_{\le x}\phi$ holds with respect to $\alpha$ if $\phi$ holds within the next~$\alpha(x)$ steps, while $\G_{\le y}\phi$ holds with respect to $\alpha$, if $\phi$ holds at least for the next $\alpha(y)$ steps. Thus, intuitively, the variables bound the scope of the operators. In particular, $\prompt$ can be seen as the fragment of $\pltl$ without parameterized always operators and where all parameterized eventually operators are parameterized by the same variable.

Alur et al.\ showed that the model checking problem for $\pltl$, where the variable valuation~$\alpha$ is existentially quantified, is $\pspace$-complete, and therefore not harder than $\ltl$ model checking. Later, a similar result was shown for solving infinite games with $\pltl$ winning conditions, which is still complete for doubly-exponential time~\cite{journals/tcs/Zimmermann13}. As for $\prompt$, distributed synthesis for $\pltl$ specifications has never been considered before. 

The second logic we consider in this section is parametric linear dynamic logic ($\pldl$)~\cite{journals/iandc/FaymonvilleZ17}, which has its roots in another shortcoming of $\ltl$: it lacks the full expressive power of the $\omega$-regular languages. There is a long line of extensions of $\ltl$ addressing this issue~\cite{LeuckerSanchez07, VardiWolper94, Wolper1983}. Most recently, Vardi introduced linear dynamic logic ($\ldl$), which adds regular expressions as \emph{guards} to the temporal operators of $\ltl$: the formula~$\ddiamond{g}\phi$ holds if there is a position such that the prefix up to it matches  the guard~$g$ and $\phi$ holds at this position. Similarly, $\bbox{g}\phi$ holds, if $\phi$ holds at all positions where the prefix up to it matches the guard. Thus, the diamond operator is a guarded eventually operator and the box operator is a guarded always operator. Vardi showed that $\ldl$ has the exponential compilation property~\cite{Vardi11}, i.e., formulas can be translated into equivalent B\"uchi automata of exponential size. Thus, $\ldl$ model checking is still $\pspace$-complete while solving $\ldl$ games is still $\twoexp$-complete.

Now, $\pldl$ is obtained by allowing parameterized diamond and box operators, with the expected semantics. For the first time, this logic addresses both shortcomings of $\ltl$, lack of timing constraints and limited expressiveness, simultaneously. Even in this setting, model checking is just $\pspace$-complete and solving games is $\twoexp$-complete~\cite{journals/iandc/FaymonvilleZ17}. Distributed synthesis for $\pldl$ specifications has never been considered before. 

In this section, we address the distributed synthesis problem for both logics, starting with the synchronous variant. For $\pltl$, we rely on a reduction to the $\prompt$ synthesis problem. The variable valuation~$\alpha$ will be existentially quantified in the problem statement, just as the bound~$k$ in the case of $\prompt$ synthesis is existentially quantified. Now, consider a parameterized always operator~$\G_{\le y}\phi$: if $\phi$ is satisfied for at last $\alpha(y)$ steps, then also for at least zero steps, i.e., at the current position. Thus, when the value for $y$ is existentially quantified, $\G_{\le y}\phi$ degenerates to the formula~$\phi$, as $y$ can always be instantiated with $0$. 

Dually, consider a parameterized eventually operator~$\F_{\le x}\phi$: if $\phi$ holds at least once within the next $\alpha(x)$ steps, then also at least once within the next~$k$ steps, for every $k \ge \alpha(x)$. Thus, if $\alpha$ is existentially quantified, then one can replace all variables parameterizing parameterized eventually operators by a unique one. By applying these two replacements, one obtains an equivalent $\prompt$ formula, provided $\alpha$ is existentially quantified. In fact, these observations were the impetus to introduce $\prompt$. However, the situation is different when one is interested in a fixed variable valuation or for optimization problems. In this case, the replacements are no longer valid.

Then, we consider the synchronous synthesis problem for $\pldl$, which we solve along the same lines as for its special case $\prompt$:  the alternating color technique has been reformulated for $\pldl$ and the exponential compilation property holds as well. Finally, we also discuss the asynchronous synthesis problem. Here, the approach for $\pltl$ and $\pldl$ is similar. Hence, we restrict our attention to the case of $\pldl$, as it subsumes $\pltl$. 
\subsection{Synchronous Distributed Synthesis for Parametric Linear Temporal Logic}

Let $\Var$ be an infinite set of variables and let $\ap$ be a set of atomic propositions. The formulas of $\pltl$ are given by the grammar
\begin{equation*}\varphi \Coloneqq a \mid \neg a \mid \varphi \wedge \varphi \mid \varphi \vee
\varphi
  \mid \X \varphi \mid \varphi \U \varphi \mid \varphi \R \varphi \mid
  \mathbf{F}_{ \le z } \varphi \mid \mathbf{G}_{ \le z} \varphi
,\end{equation*}
where $a \in \ap$ and $z \in \Var$. As before, we use the derived operators~$\F$ and $\G $ as well as implications, which are defined as for $\prompt$. 

The set of sub-formulas of a $\pltl$ formula~$\varphi$ is denoted by $\cl( \varphi )$
and we define the size of $\varphi$ to be the cardinality of $\cl(\varphi)$.
Furthermore, we define 
\[\varF( \varphi ) = \{ z\in \Var \mid \F_{\le z} \psi \in
\cl( \varphi) \}\] to be the set of variables parameterizing eventually operators in
$\varphi$, and 
\[\varG( \varphi ) = \{ z\in \Var \mid \G_{\le z} \psi \in \cl( \varphi)  \} \]
to be the set of variables parameterizing always operators in $\varphi$. Finally,
$\var( \varphi ) = \varF( \varphi ) \cup \varG( \varphi )$ denotes the set of all variables appearing in $\varphi$.

To evaluate formulas, we define a variable valuation to be a
mapping~$\alpha\colon \Var \rightarrow \nats$ mapping each variable to a value. 
Now, we can define the model
relation between a path~$w = w_0 w_1  w_2  \cdots $, a
position~$n$ of $w$, a variable valuation~$\alpha$, and a~$\pltl$ formula. 
For the atomic propositions, Boolean connectives, and standard temporal operators, it is defined as for $\prompt$, while for the parameterized operators it is defined as follows:
\begin{itemize}
%
%
%
%
%
%

\item $(w,n,\alpha)\models\F_{\le z}\varphi$ if, and only if, there exists a $j \le \alpha(z)$ such that $(w,n+j,\alpha)\models\varphi$.

\item $(w,n,\alpha)\models\G_{\le z}\varphi$ if, and only if, for every $j \le \alpha(z)$: $(w,n+j,\alpha)\models\varphi$.

\end{itemize}

For the sake of brevity, we write $(w,\alpha) \models \varphi$ instead of
$(w,0,\alpha) \models \varphi$ and say that $w$ is a model of $\varphi$ with
respect to $\alpha$. 

As usual for parameterized temporal logics, the use of variables has to be
restricted: parameterizing eventually and always operators by the same variable leads
to an undecidable satisfiability problem~\cite{journals/tocl/AlurETP01}.

\begin{definition}
\label{def_wellformedformula}
A $\pltl$ formula~$\varphi$ is well-formed if $\varF( \varphi ) \cap \varG( \varphi ) =
\emptyset$.
\end{definition}

In the following, we only consider well-formed formulas and omit the qualifier~\myquot{well-formed}. Also, we will denote
variables in $\varF( \varphi )$ by $x$ and variables in $\varG( \varphi )$ by $y$, if the formula~$\varphi$ is clear from the context.

Our solution for the $\pltl$ synthesis problem is based on the monotonicity of the parameterized temporal operators explained earlier, which is formalized in the following lemma. 

\begin{lemma}[\cite{journals/tocl/AlurETP01}]
\label{lemma_monotonicity}
Let $\varphi$ be a $\pltl$ formula and let $\alpha$ and
$\beta$ be variable valuations satisfying $\alpha ( x ) \le \beta ( x)$, for
each $x \in \varF( \varphi)$, and $\alpha ( y ) \ge \beta ( y)$, for each $y \in \varG( \varphi)$. If $(w, \alpha) \models \varphi$, then $(w, \beta) \models
\varphi$.
\end{lemma}

Thus, let $\phi$ be a $\pltl$ formula and let $\phi'$ be the $\prompt$-formula obtained from $\phi$ by inductively replacing each sub-formula~$\F_{\le x}\psi$ by $\Fp\psi$ and each sub-formula~$\G_{\le y}\psi$ by $\psi$. The following is a straightforward consequence of the previous lemma.

\begin{corollary}
\label{coro_pltl2prompt}
Let $\phi$ be a $\pltl$ formula and let $\phi'$ be defined as above.

\begin{enumerate}
	
	\item For every~$w$, if there exists a variable valuation~$\alpha$ such that $(w, \alpha) \models \phi$, then $(w,\max_{x \in \varF(\phi)}\alpha(x)) \models \phi'$.
	
	\item  For every~$w$, if there exists a bound~$k$ such that $(w,k) \models \varphi'$, then $(w,\alpha)\models\varphi$, where~$\alpha$ maps each $x \in \varF(\phi)$ to $k$ and each other variable to $0$. 
\end{enumerate}
\end{corollary}

Let $\arch = \tuple{P,\penv,\{I_p\}_{p \in P}, \{O_p\}_{p \in P}}$ be an architecture.
Here, the \emph{synchronous $\pltl$ realizability problem for $\arch$} is the problem of deciding, given a $\pltl$ formula~$\varphi$, whether there exist a variable valuation~$\alpha$ and a finite-state implementation $f_p$, for each process $p \in P^-$, such that the distributed product $\Distprod_{p \in P^-} f_p$ satisfies $\varphi$ with respect to $\alpha$, i.e., $(\Distprod_{p \in P^-} f_p,\alpha) \models \varphi$. In this case, we say that $\phi$ is realizable in $\arch$.

\begin{theorem}
\label{thm_pltlsynchsynt}
  Let $\arch$ be an architecture.
  The synchronous $\pltl$ realizability problem for $\arch$ is decidable if, and only if, $\arch$ is weakly ordered.
\end{theorem}

\begin{proof}
Fix an architecture~$\arch$. By Corollary~\ref{coro_pltl2prompt}, a given $\pltl$ formula $\phi$ is realizable in $\arch$ if, and only if, $\phi'$ as defined in the corollary is realizable in $\arch$. Thus, Corollary~\ref{coro_promptsynthesis} yields the desired result. 
\end{proof}

Also, bounded synthesis is again applicable, as we can translate the relativized PLTL formulas into universal co-B\"uchi automata.

\egroup

\subsection{Synchronous Distributed Synthesis for Parametric Linear Dynamic Logic}

As before, let $\Var$ be an infinite set of variables and let $\ap$ be the set of atomic propositions. The formulas of $\pldl$ are given by the grammar
\begin{align*}
\varphi &\Coloneqq a \mid \neg a \mid \varphi \wedge \varphi \mid \varphi \vee \varphi
  \mid \ddiamond{g} \varphi 
  \mid \bbox{g} \varphi 
  \mid \ddiamondle{g}{z} \varphi 
  \mid \bboxle{g}{z} \varphi\\
  g & \Coloneqq \phi \mid \varphi? \mid g+g \mid g \conc g \mid g^*
\end{align*}
where $a \in \ap$, $z \in \Var$, and $\phi$ ranges over propositional formulas over $\ap$.
Here, expressions of the form~$\varphi?$ are \emph{tests}, which allow us to nest operators. The sets~$\vardiamond(\varphi)$, $\varbox(\varphi)$, and $\var(\varphi)$ are defined analogously to the sets~$\varF(\varphi)$, $\varG(\varphi)$, and $\var(\varphi)$ for $\pltl$, taking sub-formulas in tests into account.

The satisfaction relation is defined, as before, among a path $w$, a position~$n$, a variable valuation $\alpha$, and a formula $\varphi$. First, let the relation~$\Rexp(g,w,\alpha) \subseteq \nats\times\nats$ contain all pairs~$(m,n) \in \nats \times \nats$ such that $w_m \cdots w_{n-1}$ matches $g$. Formally, it is defined inductively by 
\begin{itemize}
\item $\Rexp(\phi,w,\alpha) = \set{(n, n+1) \mid w_n \models \phi}$ for propositional~$\phi$,
\item $\Rexp(\varphi?,w,\alpha) = \set{(n, n) \mid (w, n, \alpha) \models \varphi}$,
\item $\Rexp(g_0 + g_1, w, \alpha) = \Rexp(g_0, w, \alpha) \cup \Rexp(g_1, w, \alpha)$,
\item $\Rexp(g_0 \conc g_1, w, \alpha) = \set{(n_0, n_2) \mid \exists n_1 \text{ s.t. }(n_0,n_1)\in \Rexp(g_0, w, \alpha) \text{ and }  (n_1, n_2) \in \Rexp(g_1, w, \alpha)}$, and
\item $\Rexp(g^*, w, \alpha) = \set{(n,n) \mid n\in\nats} \cup$ \newline$ \set{(n_0, n_{k+1}) \mid \exists n_1, \ldots, n_{k} \text{ s.t. }  (n_j, n_{j+1}) \in \Rexp(g, w, \alpha) \text{ for all } j \le k}$.
\end{itemize}

Then, for atomic formulas and Boolean connectives it is defined as for $\prompt$, while for the four temporal operators, it is defined as follows:
\begin{itemize}

\item $(w, n, \alpha) \models \ddiamond{g}\varphi$ if there exists $j \ge 0$ such that $(n, n+j) \in \Rexp(g, w, \alpha)$ and $(w, n+j, \alpha) \models \varphi$, 

\item $(w, n, \alpha) \models \bbox{g}\varphi$ if for all $j \ge 0$, with $(n, n+j) \in \Rexp(g, w, \alpha)$, we have $(w, n+j, \alpha) \models \varphi$,

\item $(w, n, \alpha) \models \ddiamondle{g}{z}\varphi$ if there exists $j \le \alpha(z)$ such that $(n, n+j) \in \Rexp(g, w, \alpha)$ and $(w, n+j, \alpha) \models \varphi$, and

\item $(w, n, \alpha) \models \bboxle{g}{z}\varphi$ if for all $j \le \alpha(z)$ with $(n, n+j) \in \Rexp(g, w, \alpha)$, we have $(w, n+j, \alpha) \models \varphi$.

\end{itemize}

Again, we restrict ourselves to well-formed formulas, i.e., those formulas~$\varphi$ with $\vardiamond(\varphi) \cap \varbox(\varphi) = \emptyset$.
With this restriction, Lemma~\ref{lemma_monotonicity} holds for $\pldl$, too. 

\begin{lemma}
\label{lemma_monotonicity_pldlc}
Let $\varphi$ be a $\pldl$ formula and let $\alpha$ and
$\beta$ be variable valuations satisfying $\alpha ( x ) \le \beta ( x)$, for
each $x \in \vardiamond( \varphi)$, and $\alpha ( y ) \ge \beta ( y)$, for each $y \in \varbox( \varphi)$. If $(w, \alpha) \models \varphi$, then $(w, \beta) \models
\varphi$.
\end{lemma}

Recall that the alternating color technique for $\prompt$ replaces every prompt-eventually operator~$\Fp \psi$ by a formula that expresses that $\psi$ holds within one color change. In $\ltl$, this is naturally expressed by two nested until operators. In $\pldl$, parameterized diamond operators, which are the analogues of prompt-eventually operators, are guarded by regular expressions, and, thus, one has to express that both the guard holds and at most one color change occurs. The simplest way to do it is to introduce a change point bounded variant of the diamond-operator (see~\cite{journals/iandc/FaymonvilleZ17}). 

Formally, we add the operator~$\ddiamondcp{\cdot}{r}$ with the following semantics:
\begin{itemize}
	\item $(w, n, \alpha) \models \ddiamondcp{g}{r}\psi$ if there exists a $j \in \nats$ s.t.\ $(n, n+j) \in \Rexp(g, w, \alpha)$, $w_n \cdots w_{n+j-1}$ contains at most one $r$-change point, and $(w, n + j, \alpha) \models \psi$.
\end{itemize}

Let $\ldlt$ be the logic obtained by disallowing parameterized operators but allowing the change point-bounded operator, whose semantics are independent of variable valuations. Hence, we drop them from our notation for the satisfaction relation~$\models$ and the relation $\Rexp$.

We need the following results from~\cite{journals/iandc/FaymonvilleZ17} which generalizes the replacement of $\pltl$ sub-formulas~$\G_{\le y}\psi$ by $\psi$ with respect to variable valuations mapping $y$ to zero. In $\pldl$, the situation is different, e.g., the formulas~$\bboxle{g}{y}\psi$ and $\psi$ are not necessarily equivalent with respect to variable valuations mapping $y$ to zero, e.g., if $r = \varphi?$ is a test. This test has to be satisfied, even if $\alpha(y) = 0$. However, one can easily simplify the guard~$g$ to a guard~$\widehat{g}$ that captures~$g$ when restricted to matchings of length zero. 

\begin{lemma}[\cite{journals/iandc/FaymonvilleZ17}]
	\label{lemma_removeboxes}
For every $\pldl$ formula~$\varphi$ there is an efficiently constructible $\pldl$ formula~$\varphi'$ without paramterized box operators whose size is at most the size of $\varphi$ such that
\begin{enumerate}
	
	\item $\vardiamond(\varphi) = \vardiamond(\varphi')$,
	
	\item for each $\alpha$ and each $w$,  $(w, \alpha) \models \varphi$ implies $(w, \alpha) \models \varphi'$, and
	
	\item for each $\alpha$ and each $w$, $(w, \alpha) \models \varphi'$ implies $(w, \alpha_0) \models \varphi$.	

\end{enumerate}
In the third item, $\alpha_0$ is the valuation mapping each $x \in \vardiamond(\varphi)$ to $\alpha(x)$ and each other variable to $0$.
\end{lemma}

Note that the formulas~$\varphi$ and $\varphi'$ as above are equivalent, if the variable valuation is existentially quantified.

Now, given such a $\pldl$ formula~$\varphi$, let $\rel_r(\varphi)$ denote the formula obtained from the formula~$\varphi'$ as in Lemma~\ref{lemma_removeboxes} by inductively replacing each sub-formula~$\ddiamondle{g}{x}\psi$ by $\ddiamondcp{g}{r}\psi$. Furthermore, let $\alt_r = \bbox{\true^*}\ddiamond{\true^*}r \wedge \bbox{true^*}\ddiamond{\true^*}\neg r$, which is equivalent to the $\ltl$ formula~$\GF r \wedge \GF \neg r$ from above. Now, define $c_r(\varphi) = \rel_r(\varphi) \wedge \alt_r$, which is an $\ldlt$ formula.

\begin{lemma}[\cite{journals/iandc/FaymonvilleZ17}]
\label{lemma_alternatingcolor_pldl}
Let $\varphi$ be a $\pldl$ formula and let $w \in \left( \pow{\ap} \right)^{ \omega }$.
\begin{enumerate}
  \item \label{lemma_alternatingcolor_pldl2ldlt}
If $(w,\alpha)\models \varphi$, then $w' \models c_r(\varphi)$ for every $k$-spaced $r$-coloring~$w'$ of~$w$, where $k = \max_{x \in \vardiamond(\varphi)}\alpha(x)$.
  
  \item \label{lemma_alternatingcolor_ldlt2pldl}
If $w'$ is a $k$-bounded $r$-coloring of $w$ with $w' \models c_r(\varphi)$, then $(w,\alpha)\models\varphi$, where $\alpha$ maps each $x \in \vardiamond(\varphi)$ to $2k$ and each other variable to zero. 
\end{enumerate}
\end{lemma}

Finally, the exponential compilation property holds for $\ldlt$ as well: every $\ldlt$ formula can be translated into an equivalent non-determinstic B\"uchi automaton of exponential size~\cite{journals/iandc/FaymonvilleZ17}. 

Now, the (synchronous) $\pldl$  distributed  synthesis problem is defined as its analogue for $\pltl$. Let $\arch = \tuple{P,\penv,\{I_p\}_{p \in P}, \{O_p\}_{p \in P}}$ be an architecture.
Then, the \emph{synchronous $\pldl$ realizability problem for $\arch$} is the problem of deciding, given a $\pldl$ formula $\varphi$, whether there exist a variable valuation~$\alpha$ and a finite-state implementation $f_p$ for each process $p \in P^-$, such that the distributed product $\Distprod_{p \in P^-} f_p$ satisfies $\varphi$ with respect to $\alpha$, i.e., $(\Distprod_{p \in P^-} f_p,\alpha) \models \varphi$. In this case, we say that $\varphi$ is realizable in $\arch$.

\begin{theorem}
\label{thm_pldlsynchsynt}
  Let $\arch$ be an architecture.
  The synchronous $\pldl$ realizability problem for $\arch$ is decidable if, and only if, $\arch$ is weakly ordered.
\end{theorem}

\begin{proof}
Theorem~\ref{thm_synchr_prompt2ltl} holds for $\pldl$ as well, using the same proof: a $\pldl$ formula $\varphi$ is realizable in $\arch$ if, and only if, $c_r(\varphi)$ is realizable in $\arch^r$. Now, the information fork criterion holds for $\omega$-regular conditions as well~\cite{conf/lics/FinkbeinerS05}, which finishes the proof.
\end{proof}

Also, bounded synthesis is again applicable, as we can also translate the relativized PLDL formulas into universal co-B\"uchi automata.

\subsection{Asynchronous Distributed Synthesis for PLDL}

Finally, we consider the asynchronous setting. We focus on $\pldl$, as $\pltl$ is a fragment of $\pldl$ and the approach for both problems is similar.

As for the asynchronous $\prompt$ realizability problem, we require the implementations to only change their state if they are scheduled. Here, a $\pldl$ assume-guarantee specification~$\tuple{\varphi,\psi}$ consists of a pair of $\pldl$ formulas.
The asynchronous $\pldl$ assume-guarantee realizability problem asks, given an asynchronous architecture~$\arch^*$ and $\tuple{\varphi,\psi}$ as above, whether there exists a finite-state implementation~$f_p$, for each process $p \in P^-$, such that for each variable valuation~$\alpha$ there is a variable valuation~$\beta$ such that for each $w \in \Distprod_{p \in P^-} f_p$, we have that $(w, \alpha) \models \varphi$ implies $(w, \beta) \models \psi$. In this case, we say that $\Distprod_{p \in P^-} f_p$ satisfies $\tuple{\varphi,\psi}$.

To solve the problem, we use the framework of bounded synthesis and emptiness checking for B\"uchi graphs as  presented for $\prompt$ in Section~\ref{sec:asynchronous_distributed_synthesis}. In particular, we adapt the notation introduced in Subsection~\ref{sec:semi-algorith-ag}, e.g., the product system~$\tsys = \Distprod_{p \in \pminus} \tsys_p$. Our semi-decision procedure again guesses implementations and then model checks whether their product~$\tsys$ satisfies the assume-guarantee specification, based on a characterization in terms of $\tsys$ being pumpable non-empty. To this end, we have to lift Lemma~\ref{lemma_pumpabilitycharac} to $\pldl$, which again requires to remove parameterized box operators. Once more, we rely on monotonicity, but due to the quantifier alternation and the implication between $\varphi$ and $\psi$, the application is not completely trivial. Given the assumption~$\varphi$, let $\varphi'$ be the formula as described in Lemma~\ref{lemma_removeboxes}, which has no parameterized box operators. The formula~$\psi'$ is defined similarly.

\begin{lemma}
Let $\tsys$, $\varphi'$, and $\varphi'$ as above. Then, $\tsys$ satisfies $\tuple{\varphi,\psi}$ if, and only if, $\tsys$ satisfies $\tuple{\varphi',\psi'}$.
\end{lemma} 

\begin{proof}
Let $f$ denote the strategy generated by $\tsys$. 

For the implication from left to right, let $\tsys$ satisfy $\tuple{\varphi, \psi}$, i.e., for each $\alpha$ there is a $\beta$ such that for all $w \in f$: $(w, \alpha) \models \varphi$ implies $(w, \beta) \models \psi$. As $\beta$ depends on $\alpha$, we write $\beta(\alpha)$ to make the dependency clear.

Now, given some arbitrary $\alpha$ let $\alpha_0$ denote the variable valuation mapping each $x \in \vardiamond(\varphi) = \vardiamond(\varphi')$ to $\alpha(x)$ and each other variable to $0$. We claim that $(w, \alpha) \models \varphi'$ implies $(w, \beta(\alpha_0)) \models \psi'$ for all $w \in f$, which implies that $\tsys$ satisfies $\tuple{\varphi',\psi'}$.

Thus, assume the assumption is satisfied, i.e., $(w, \alpha) \models \varphi'$. Then, we also have $(w, \alpha_0) \models \varphi$ by Lemma~\ref{lemma_removeboxes}. Thus, $(w, \beta(\alpha_0)) \models \psi$, which implies $(w, \beta(\alpha_0)) \models \psi'$, again by Lemma~\ref{lemma_removeboxes}.

For the other implication, let $\tsys$ satisfy $\tuple{\varphi', \psi'}$, i.e., for each $\alpha$ there is a $\beta$ such that for all $w \in f$: $(w, \alpha) \models \varphi'$ implies $(w, \beta) \models \psi'$. Again, as $\beta$ depends on $\alpha$, we write $\beta(\alpha)$ to make the dependency clear.

We claim that $(w, \alpha) \models \varphi$ implies $(w, \beta(\alpha)) \models \psi$ for all $w \in f$, which implies that $\tsys$ satisfies $\tuple{\varphi,\psi}$.

Thus, assume the assumption is satisfied, i.e., $(w, \alpha) \models \varphi$. Then, we also have  $(w, \alpha)  \models \varphi'$ by Lemma~\ref{lemma_removeboxes}. Thus, $(w, \beta(\alpha)) \models \psi'$, which implies  $(w, (\beta(\alpha))_0) \models \psi$, again by Lemma~\ref{lemma_removeboxes}. Here, $(\beta(\alpha))_0$ maps each variable in $\vardiamond(\psi) = \vardiamond(\psi')$ to $(\beta(\alpha))(x)$ and each other variable to $0$.
\end{proof}

To simplify the notation we can assume that $\varphi$ and $\psi$ do not contain any parameterized box operators. Thus, the alternating color technique is applicable to them. Also, there is a non-deterministic B\"uchi  automaton $\nbw_{\overline{c}_{r'}(\psi) \land c_r(\varphi)} = \tuple{2^{I \cup O  \cup \set{r,r'}}, Q, q_0, \delta, \bwin}$, where $\overline{c}_{r'}(\psi) = \alt_{r'} \land \neg\rel_{r'}(\psi)$ whose language contains exactly those paths that satisfy ${\overline{c}_{r'}(\psi) \land c_r(\varphi)}$~\cite{journals/iandc/FaymonvilleZ17}. Then, Lemma~\ref{lemma_pumpabilitycharac} holds in this setting as well.

\begin{lemma}
\label{lemma_pumpabilitycharac}
	Let $\tuple{\varphi,\psi}$ be a PLDL assume-guarantee specification, $\arch^*$ be an asynchronous architecture, and $\tsys_p$ be a finite-state implementation for each system process $p \in \pminus$.
	The distributed product $\tsys = \Distprod_{p \in \pminus} \tsys_p$ does not satisfy $\tuple{\varphi,\psi}$ if, and only if, the product of $\tsys$ and $\nbw_{\overline{c}_{r'}(\psi) \land c_r(\varphi)}$ is pumpable non-empty.
\end{lemma}

From here on the algorithm is similar to that described in Section~\ref{sec:asynchronous_distributed_synthesis} and we obtain the same semi-decidability result.

\begin{corollary}
  Let $\arch^*$ be an asynchronous architecture.
  The asynchronous PLDL assume-guarantee realizability problem for $\arch^*$ is semi-decidable.
\end{corollary}

\bgroup
\renewcommand{\phi}{\varphi}

\section{Conclusion}\label{sec:conc}

In this paper, we have initiated the investigation of distributed synthesis for parameterized specifications, in particular for $\prompt$, PLTL, and PLDL. These logics subsume $\ltl$, and additionally allow to express bounded satisfaction of system properties, instead of only eventual satisfaction. To the best of our knowledge, this is the first treatment of parametrized temporal logic specifications in distributed synthesis.

We have shown that for the case of synchronous distributed systems, we can reduce the $\prompt$ synthesis problem to an $\ltl$ synthesis problem. Thus, the complexity of $\prompt$ synthesis corresponds to the complexity of $\ltl$ synthesis, and the $\prompt$ realizability problem is decidable if, and only if,  the $\ltl$ realizability problem is decidable. For the case of asynchronous distributed systems with multiple components, the $\prompt$ realizability problem is undecidable, again corresponding to the result for $\ltl$. For this case, we give a semi-decision procedure based on a novel method for checking emptiness of two-colored B\"uchi~graphs. Finally, we have shown that all these results also hold for $\pltl$ and $\pldl$. Furthermore, the approach is also applicable to $\pltl$ and $\pldl$ in a weighted setting~\cite{Zimmermann2016}, as even these logics have the exponential compilation property and the alternating color technique is applicable to them as well. Finally, we conjecture that the approach also extends to assume-guarantee synthesis with mutual assumptions between different processes~\cite{ChatterjeeH07,BloemCJK15}.

Among the problems that remain open is realizability of $\prompt$ specifications in asynchronous distributed systems with a single component. This problem can be reduced to the (single-process) assume-guarantee realizability problem for $\prompt$, which was left open in~\cite{journals/fmsd/KupfermanPV09}.

In the future, we also want to look into the synthesis of distributed systems with a parametric number of components~\cite{JacobsB14,KhalimovJB13} from parameterized temporal logics. In addition to the even more difficult problems of realizability and synthesis, new questions arise in this context, such as: how does the bound on prompt eventualities increase with the number of components in the system?

\egroup

\bibliographystyle{elsarticle-num}
\bibliography{main}

\end{document}